\documentclass[11pt,a4paper]{article}
\usepackage{amsfonts}
\usepackage[T1]{fontenc}
\usepackage[style=ieee, citestyle=numeric-comp, backend=biber]{biblatex}
\usepackage{graphics}
\usepackage{epsfig}
\usepackage{enumerate}
\usepackage{hyperref}
\usepackage{mathtools}
\mathtoolsset{showonlyrefs}
\usepackage{amsmath}
\usepackage{graphicx}
\usepackage{enumerate}
\usepackage{algorithm}
\usepackage{algpseudocode}
\usepackage{amssymb}
\usepackage{bbm}
\usepackage{color}
\usepackage{amsthm}
\usepackage{physics}
\usepackage{caption}
\usepackage{subcaption}
\usepackage{interval}
\usepackage{multirow}
\intervalconfig{
soft open fences
}
\usepackage{tikz}
\usetikzlibrary{quantikz}
\usepackage{adjustbox}
\usepackage{svg}
\usepackage{authblk}

\usepackage[left=1in,right=1in,top=1in,bottom=1in]{geometry}
\usepackage{abstract}
\newcommand{\appref}[1]{\hyperref[#1] {{Appendix~\ref*{#1}}}}
\newcommand{\be}{\begin{eqnarray} \begin{aligned}}
\newcommand{\ee}{\end{aligned} \end{eqnarray}}
\newcommand{\benn}{\begin{eqnarray*} \begin{aligned}}
\newcommand{\eenn}{\end{aligned} \end{eqnarray*}}

\newcommand*{\cA}{\mathcal{A}} 
\newcommand*{\cB}{\mathcal{B}}
\newcommand*{\cC}{\mathcal{C}}

\newcommand*{\cL}{\mathcal{L}}
\newcommand*{\cM}{\mathcal{M}}
\newcommand*{\cN}{\mathcal{N}}

\newcommand*{\cQ}{\mathcal{Q}}

\newcommand*{\cP}{\mathcal{P}}

\newcommand*{\cU}{\mathcal{U}}

\newcommand*{\supp}{\mathrm{supp}}

\newcommand{\bc}{\begin{center}}
\newcommand{\ec}{\end{center}}

\newtheorem{theorem}{Theorem}[section]
\newtheorem{lemma}[theorem]{Lemma}

\newtheorem{claim}[theorem]{Claim}

\newtheorem{corollary}[theorem]{Corollary}

\usepackage{amsfonts}

\def\01{\{0,1\}}
\newcommand{\ceil}[1]{\lceil{#1}\rceil}

\newcommand*{\poly}{\mathsf{poly}}

\definecolor{lightblue}{RGB}{109, 194, 247}
\definecolor{bulklogical}{RGB}{26,11,11}
\definecolor{bulk}{RGB}{217, 212, 212}
\definecolor{lightred}{RGB}{255, 148, 148}
\definecolor{mygreen}{RGB}{0, 100, 0}

\title{How to fault-tolerantly realize\\
 any quantum circuit with local operations}
\author[1,2]{Shin Ho Choe}
\author[1,2]{Robert König}
\affil[1]{School of Computation, Information and
Technology, Technical University of Munich}
\affil[2]{Munich Center for Quantum Science and Technology (MCQST), Munich, Germany}
\date{\today}

\begin{document}
\maketitle

\begin{abstract}
We show how to realize a general quantum circuit involving gates between arbitrary pairs of qubits   by means of geometrically local quantum  operations and efficient classical computation. 
We prove that circuit-level local stochastic noise modeling an imperfect implementation of our derived schemes  is equivalent to local stochastic noise in the original circuit. Our constructions
incur a constant-factor increase in the quantum circuit depth and a polynomial overhead in the number of qubits: To execute  an arbitrary quantum circuit  on $n$~qubits, we give a 3D quantum fault-tolerance architecture involving $O(n^{3/2} \log^3 n)$ qubits, and a quasi-2D architecture  using~$O(n^2 \log^3 n)$ qubits. Applied to recent fault-tolerance constructions, this gives a fault-tolerance threshold theorem for universal quantum computations with local operations, a polynomial qubit overhead and a quasi-polylogarithmic depth overhead. More generally, our transformation dispenses with the need for considering the locality of operations when designing schemes for fault-tolerant quantum information processing.
\end{abstract}

\section{Introduction}
Throughout quantum complexity theory, quantum algorithms design and quantum fault-tolerance, the notion of a quantum operator or operation acting only on a small (i.e., constant) number of qubits is ubiquitous. For example, the complexity class QMA~\cite{knillQuantumRandomnessNondeterminism1996, kitaevClassicalQuantumComputation2002, marriottQuantumArthurMerlinGames2004}, the natural analog of NP, has been shown to have a natural complete problem, the $k$-local Hamiltonian problem~\cite{kitaevClassicalQuantumComputation2002}. Here $k$-local refers to the fact that the considered Hamiltonians involve terms acting on arbitrary subsets of a constant number $k$~of qubits.  This notion of locality has occasionally been referred to as Kitaev-local to distinguish it from more commonly used geometric locality notions arising, e.g., in physics. In quantum compiling, $2$-qubit entangling unitaries such as the CNOT gate play a special role: together with general single-qubit unitaries, such a gate provides computational universality when it can be applied to arbitrary pairs of qubits. In fault-tolerance theory, special attention is paid to low-weight errors as formalized by the notion of code distance, and quantum low-density parity check (LDPC) codes are of particular interest because their syndrome measurements involve only a constant number of qubits (see Ref.~\cite{breuckmannQuantumLowDensityParityCheck2021} for a recent review).

While Kitaev-local quantum operation provide a versatile abstraction, considerations relating to real-world experiments are typically encumbered by more stringent locality constraints arising from technological limitations. Corresponding geometric notion of locality typically involve an interaction graph formalizing what interactions or operations are available. With an often considerable amount of effort, such geometric locality notions can sometimes be taken into account successfully: This has led, e.g., to proofs of QMA-completeness of the local Hamiltonian problem on certain lattices of qubits~\cite{AharonovGottesmanIraniKempe09}. Similarly, circuit compilation taking into account the interaction graph of an underlying device is an increasingly relevant topic as larger devices are becoming available (see e.g., Ref.~\cite{bapat03}). In the area of fault-tolerance, the fundamental question of whether noisy operations permit universal fault-tolerant quantum computation was addressed and answered positively early on,  see~\cite{aharonovFaulttolerantQuantumComputation1997, Gottesman_2000, svoreLocalftqc, svoreNoisethreshold2007}. In a similar vein, the search for and use of codes with geometrically local stabilizers has been major theme for the last few decades. 

While the general program of rendering Kitaev-local constructions geometrically local has been successful in many instances, there are cases where it meets fundamental obstructions. For example, this is captured by trade-off-relations between the number of physical and logical qubits, the code distance and the energy barrier (see e.g.,~\cite{BravyiTerhal2009,BravyiPoulinTerhal10,Haah21}), or results on the (in)existence of protected logical gates in geometrically local stabilizer codes (see e.g.,~\cite{brakoetopologicalprot,PastawskiYoshida15,Kubica_2015}). Given these limitations, one may ask to what extent it is reasonable to afford oneself the luxury of thinking only about Kitaev-locality when dealing with the real world.

\subsection*{Our contribution}
Here we argue that -- in a very general sense -- it is in fact possible to meaningfully separate quantum computational questions from geometric locality considerations. 
Let us consider a general adaptive quantum circuit~$\cQ$. Here adaptivity means that every basic operation of the circuit~$\cQ$ such as a single- or two-qubit quantum gate can be determined by a (classically) efficiently computable function from some bits of measurement outcomes associated with previously applied measurements. (We refer to Section~\ref{sec:localization} for a precise definition.)
To this end, we propose a transformation which takes as input (a description of) a general quantum circuit~$\cQ$ involving non-local gates, and produces (a description of) a recompiled ``localized'' quantum circuit~$\cQ'$ with the following properties:
\begin{enumerate}[(i)]
\item\label{it:geometricallylocalproperty}
  The circuit~$\cQ'$ is  composed of geometrically local gates: Every two-qubit operation acts on a pair of qubits whose distance is upper bounded by a constant, and all qubits are arranged on a regular rectangular lattice in 3D.
We give two constructions: one using a slab-like architecture of linear dimensions $L\times L\times \Theta(L)$, and one using a quasi-2D architecture of dimensions $L\times L\times O(\log L)$, i.e., a thickened square shape.
 
\item\label{it:overheadpolylogarithmic}
The circuit~$\cQ'$ uses $O(n^\alpha \log^3 n)$ qubits, where $\alpha=3/2$ or $\alpha= 2$ for the 3D- respectively quasi-2D case. Here $n$ is the number of qubits involved in the original circuit~$\cQ$.

\item\label{it:circuitefficiency}
  The (quantum) circuit depth of~$\cQ'$ is identical to that of~$\cQ$. Compared to~$\cQ$, the circuit~$\cQ'$ additionally involves efficient, i.e., polynomial-time classical processing.

\item\label{it:functionalityimplementation}
The circuit~$\cQ'$ realizes the same functionality as~$\cQ$. In particular, by taking a suitable marginal (i.e., ignoring certain output bits), the  output distribution of~$\cQ'$ simulates that of~$\cQ$ exactly.
   \end{enumerate}
   We note that Properties~\eqref{it:geometricallylocalproperty}--\eqref{it:functionalityimplementation} are relatively straightforward to achieve using entanglement swapping and suitable routing protocols as discussed below
    (see Section~\ref{sec:localization}). The key new property we establish for~$\cQ'$ is the following, see below for a formal definition of local stochastic noise and noisy implementations:

   %\begin{enumerate}%[(i)]\setcounter{enumi}{4}
   %\item\label{it:noisycircuiterrorpropagation}
   \begin{theorem}\label{thm:noisycircuiterrorpropagation}
   A noisy implementation of~$\cQ'$ with local stochastic noise of strength~$p$ is equivalent to a noisy implementation of~$\cQ$ with local stochastic noise of
   strength~$Cp^c$ for some constants~$C,c>0$. That is, the corresponding output distribution is identical to that obtained by running~$\cQ$ with noise (of a related strength).
   \end{theorem}
  
   Taken together,     Properties~\eqref{it:geometricallylocalproperty}--\eqref{it:functionalityimplementation} and Theorem~\ref{thm:noisycircuiterrorpropagation} imply  that noise-tolerance considerations of a quantum computational scheme may be studied entirely in the context of general circuits, i.e., quantum circuits involving Kitaev-local (but not necessarily geometrically local) gates. This leads to a considerable simplification in settings where the underlying structure is sparse but does not necessarily obey operationally meaningful geometric locality constraints.

As an application, we consider the problem of approximately sampling from the output distribution of an ideal circuit by means of a circuit made of noisy components.
Here we show that this can be achieved by a 3D-local 
 adaptive circuit~$\cQ_{\mathrm{ideal}}$ with a polynomial qubit-overhead and quasi-polylogarithmic time-overhead. That is, we have the following:
\begin{corollary}[3D-local fault-tolerant quantum computation]\label{cor:faulttolerancelocal}
There is a threshold~$p_0>0$ on error strength such that the following holds for all sufficiently large~$n$ and an arbitrary constant~$\varepsilon \in (0,1)$. 
Let~$\cQ_{\mathrm{ideal}}$ be an adaptive quantum circuit on $n$ qubits of quantum depth~$T(n)=O(\poly(n))$.
There is a circuit~$\cQ'$ with the following properties:
\begin{enumerate}[(i)]
    \item 
    The circuit~$\cQ'$ uses $O(n^{3/2}\log^3 n)$ qubits and is local when these are arranged on a 3D grid graph.
    \item The quantum depth of~$\cQ'$ is of order $T(n) \cdot \exp(O(\log^2(\log(n/\varepsilon))))$.
    \item 
    Any noisy implementation of~$\cQ'$ with local stochastic noise of strength~$p<p_0$ produces a sample from a distribution whose total variation distance to the output distribution of~$\cQ_{\mathrm{ideal}}$ is upper bounded by~$\varepsilon$. 
\end{enumerate}
\end{corollary}
\noindent  
We discuss the derivation of this result in Section~\ref{sec:ftquantumcomputation}. Thanks to our general transformation, it is an immediate consequence of the recent work by Yamasaki and Koashi~\cite{yamasakiTimeEfficientConstantSpaceOverheadFaultTolerant2024}. These authors demonstrate how to achieve time-efficient constant-qubit overhead fault-tolerant quantum computation by using a concatenation of multiple small-size quantum codes (albeit with non-local gates). We note that  -- alternatively -- our transformation can also be applied directly to other schemes for fault-tolerant quantum computation such as those relying on constant-rate quantum LDPC codes~\cite{gottesmanConstantOverhead,fawziConstantOverheadQuantum2018}. However, this leads to a polynomial quantum circuit depth overhead resulting from the sequential application of two-qubit gates in these schemes. It is currently not known whether this depth overhead can be reduced.

What sets our construction apart from earlier work on fault-tolerant quantum computation with local operations~\cite{aharonovFaulttolerantQuantumComputation1997, Gottesman_2000, svoreLocalftqc, svoreNoisethreshold2007} is  its limited quantum circuit-depth overhead: Earlier works rely on permuting qubits to achieve locality and consequently incur a polynomial circuit-depth overhead compared to the original ideal (non-local) circuit. In contrast, applying our construction to a (non-local) fault-tolerant quantum circuit incurs only constant circuit-depth-overhead. This results in the quasi-polylogarithmic quantum circuit scaling when applied to the protocol of~\cite{yamasakiTimeEfficientConstantSpaceOverheadFaultTolerant2024}. 

Let us point out an important difference 
between our result and the  construction of (non-local) fault-tolerant quantum circuits in~\cite{yamasakiTimeEfficientConstantSpaceOverheadFaultTolerant2024}. In the latter work, idling (wait)  operations in a quantum circuit required to wait for the results of classical computation (e.g., for decoding or gate teleportation) are explicitly taken into account (i.e., the corresponding notion of circuit depth includes such operations). In contrast, our construction
requires additional classical computations for every timestep of the original quantum circuit. Correspondingly, our notion of ``quantum depth'' only involves the number of timesteps, but not the additional classical (efficient) computation. In our construction, these computations involve  finding a minimal matching in each of  $O(n)$ slab-like graphs of linear dimension~$O(n) \times O(\log n) \times O(\log n)$, which is achievable in polynomial time by Edmond's algorithm~\cite{Edmonds1965}.
We leave as a future work to check whether the construction of Ref.~\cite{yamasakiTimeEfficientConstantSpaceOverheadFaultTolerant2024} is fault-tolerant even when accounting for associated waiting times.

\subsection*{Related work}
Most existing constructive work for fault-tolerance with geometrically local operations considers the design of quantum memories. This can be seen as a special case of fault-tolerant computation where the identity operation is approximated using local (noisy) operation.

Given an encoded state in a quantum error-correcting code, one can maintain the logical information for a fixed duration by repeatedly measuring syndromes and applying a suitable recovery map (or ``decoding'') involving a classical computation on the syndrome bits. Such protocols only involve local (quantum) processing if a quantum error-correcting code whose stabilizers are geometrically local in Euclidean space is used. Let us mention three of the most important approaches towards storing multiple logical qubits in a quantum memory.

A first approach is to use multiple quantum error-correcting codes in parallel, with each code encoding a fixed number of logical qubits. 
For example, multiple tiles of 2D surface codes can be used for this purpose~\cite{denniskitaevlandalpreskill}. One drawback of using 2D surface codes for this purpose is that their decoding is not single-shot: In the case where measurements are imperfect,  the decoding needs a sequence of  (noisy) syndromes  obtained from repeated measurements. In three spatial dimensions, single-shot decoding is possible with the  3D~subsystem toric codes~\cite{kubicaSingleshotQuantumError2022}.

A second approach is to use quantum LDPC codes with geometrically local stabilizer generators. It is known from seminal work by Bravyi and Terhal~\cite{BravyiTerhal2009}, Bravyi, Poulin and Terhal~\cite{BravyiPoulinTerhal10}, as well as Haah~\cite{Haah21} that there are fundamental limitations to quantum codes with geometrically local stabilizers. These rule out the existence of good (i.e., constant-rate) quantum stabilizer codes with geometrically local generators in low spatial dimensions.
In more detail, an $[[n,k,d]]$-code with a constant density of qudits in~$\mathbb{R}^D$ and local stabilizer generators necessarily obeys the bounds
  \begin{align}
 d &\in O(n^{\frac{D-1}{D}})\qquad \ \  \textrm{see Ref.~\cite{BravyiTerhal2009}}\label{eq:btbound} \\
 kd^{\frac{2}{D-1}}&\in O(n)\qquad\qquad \textrm{see Ref.~\cite{BravyiPoulinTerhal10}}\label{eq:bptbound}\\
 k & \in O(n^{\frac{D-2}{D}})\qquad\textrm{ see Ref.~\cite{Haah21}}\ .\label{eq:haahbound}
 \end{align}
 It is natural to try to construct quantum codes coming close to or saturating the trade-off bounds~\eqref{eq:btbound}--\eqref{eq:haahbound}. This program has been quite successful recently, resulting in several constructions derived from quantum LDPC codes: Portnoy~\cite{portnoy2023local} has constructed codes with optimal scaling of the number of logical qubits and the code distance, up to polylogarithmic factors. Subsequent work by Lin et al.~\cite{lin2023geometrically} gave constructions additionally saturating the upper bound of~\cite{BravyiTerhal2009} on the energy barrier. Concurrent work by Baspin and Williamson~\cite{baspinwilliamson23} provides a particularly elegant construction of optimal codes in $D=3$~dimensions: These are $[[n,\Theta(n^{1/3}), \Theta(n^{2/3})]]$-codes on a rectangular lattice of linear extent~$n^{1/3}$. Note, however, that the question of how to construct and analyze protocols for fault-tolerant computation using such codes remains largely open. In particular, no analytic fault-tolerance threshold results have been established to date.

A third approach is to go beyond
codes with geometrically local stabilizer generators, considering either codes with high-weight stabilizers or LDPC codes whose stabilizer generators are not geometrically local when embedded in~$\mathbb{R}^D$. In either case, this necessitates an analysis of the propagation of errors in local (typically larger-depth) circuits designed to efficiently realize non-local ones. 
 This is the approach pursued in pioneering works on fault-tolerance with local operations~\cite{aharonovFaulttolerantQuantumComputation1997, Gottesman_2000, svoreLocalftqc, svoreNoisethreshold2007}. Following this idea, Pattison et al.~\cite{pattisonHierarchicalMemoriesSimulating2023} introduce so-called hierarchical code encoding~$\Theta(n)$-logical qubits using~$\Theta(n \log^2 n)$ qubits (with $O(\log n)$-weight stabilizer generators).
The authors showed that measuring all syndromes is possible with a local quantum circuit with depth~$O(\sqrt{n} \log n)$ where all qubits are aligned in a bilayer structure.
Based on this code and its syndrome measurement circuit, they construct a 2D-local quantum memory for~$n$ logical qubits using only~$\Theta(n \log^2 n)$ qubits.

Our construction also essentially follows this last approach, but encompasses both the construction of local memories as well as fault-tolerant computation. As described above (see Theorem~\ref{thm:noisycircuiterrorpropagation}), our recompilation procedure renders any quantum circuit~$\cQ$ geometrically local while preserving its quantum depth (up to a constant factor). When applied to a syndrome extraction or recovery circuit associated with an LDPC code, this realizes low-weight (but geometrically non-local) measurements in constant depth, and yields a local quantum memory. When applied to (non-local) fault-tolerance constructions based on 
concatenated codes (with high-weight stabilizers), our construction provides a threshold theorem for fault-tolerant computation with local operations, which is depth-efficient (see Corollary~\ref{cor:faulttolerancelocal}).

\subsubsection*{Converse bounds and optimality}
For local quantum memories following the third approach, there are known trade-off relations between the qubit overhead and the circuit depth. Delfosse, Beverland and Tremblay~\cite{delfosseBoundsStabilizerMeasurement2021} established a trade-off relation between the depth of a syndrome extraction circuit and the qubit overhead when using  constant-rate quantum LDPC codes whose Tanner graphs are expanders.
Their bound applies to geometrically local syndrome extraction circuits that are Clifford circuits augmented by Paulis controlled by parities of measurement outcomes. (This allows, in particular, for e.g., teleportation within the circuit.) In more detail, 
 consider
a constant-rate~($R$) LDPC code encoding
$n$ logical qubits into $n/R$~qubits. Then any such Clifford circuit for syndrome extraction, which uses $n^{\mathrm{tot}}\geq n/R$~qubits and is geometrically local in~$\mathbb{R}^D$, must have quantum circuit depth~$T$ obeying
\begin{align}\label{eq:dbtbound}
    T  = \Omega\left(n/(n^{\mathrm{tot}})^{1-1/D}\right) \ 
  \qquad \text{see \cite[Corollary 2]{delfosseBoundsStabilizerMeasurement2021}.}
\end{align}
One of our intermediate constructions saturates this bound when applied to an ideal (non-local) syndrome extraction circuit~$\cQ$ associated with a constant-rate LDPC codes. (This includes, in particular, expander codes as assumed in the bound~\eqref{eq:dbtbound}.) Because of the LDPC property, we can assume that~$\cQ$ is a constant-depth circuit on $\Theta(n)$~qubits -- here we take into account one auxiliary qubit  for measuring each constant-weight stabilizer generator.  
To arrive at the fault-tolerant circuit~$\cQ'$ discussed in Theorem~\ref{thm:noisycircuiterrorpropagation}, our construction involves the construction of an intermediate circuit~$\cQ''$ for syndrome extraction with the following properties: The circuit~$\cQ''$ is $3D$-local, uses~$n^{\mathrm{tot}}=O(n^{3/2})$ qubits and has depth~$T=O(1)$ (see Lemma~\ref{lem:localizingcircuit} below). With these parameters, the circuit~$\cQ''$ saturates the bound~\eqref{eq:dbtbound}. 
We note, however, that this circuit~$\cQ''$ is not fault-tolerant. 
(Our final construction~$\cQ'$ has the fault-tolerance property, and its parameters saturate the bound~\eqref{eq:dbtbound} up to polylogarithmic factors, but the bound does not apply in this case because our construction uses controls that are not simply parities of measurement results.)

Our construction focuses on $T=O(1)$, i.e., constant circuit depth. In contrast, the hierarchical code~\cite{pattisonHierarchicalMemoriesSimulating2023} constitutes a construction in a different regime. It includes, in particular, instances where 
\begin{align}
    T=\Theta(\sqrt{n}) \quad \textrm{and} \quad n^{\mathrm{tot}} = \Theta(n)
    \qquad \textrm{or} \qquad     
    T=\Theta(\sqrt{n}\log n) \quad \textrm{and} \quad n^{\mathrm{tot}} = \Theta(n\log^2 n) \ ,
\end{align}
respectively (depending on the size of the surface codes used in the code concatenation step). The former case saturates the bound~\eqref{eq:dbtbound}, but does not provide fault-tolerance. (Similar to our construction of~$\cQ'$, the latter case provide a fault-tolerant construction whose  parameters saturate~\eqref{eq:dbtbound}. However, the bound does not apply in this case because the code is not an LDPC code as required by the bound.)

We emphasize that minimizing syndrome extraction circuit depth is not the only relevant aspect when considering quantum fault-tolerance: It is equally important to have limited error-propagation. 
Here the bound of  Baspin, Fawzi and Shayeghi~\cite{baspinLowerBoundOverhead2023} provides some insight as it takes into account the logical error probability: They establish a trade-off relation between qubit overhead and the circuit depth of a recovery map for a local quantum memories built on general quantum error-correcting codes.
To discuss their bound, let us specialize it to the case where $n$~logical qubits are considered and
an exponentially small logical failure probability of at most~$\exp(-c\cdot n)$ is assumed. According to their bound, if there is a recovery map achieving this, which is geometrically local in~$\mathbb{R}^D$,  uses~$n^{\mathrm{tot}}$ qubits and has quantum  circuit depth~$T$, then
\begin{align}\label{eq:bfsbound}
    T= \Omega\left(\frac{n^{1+1/D}}{n^{\mathrm{tot}}}\right) \  \qquad \text{see \cite[Theorem 28]{baspinLowerBoundOverhead2023}} . 
\end{align}
We note that our construction provides a local quantum memory as follows: We start with an adaptive circuit~$\cQ$ which consists in syndrome measurement followed by single-shot decoding (Pauli correction) for a good quantum LDPC code (encoding $n$~qubits). Then the resulting fault-tolerant localized circuit~$\cQ'$ is $3D$-local, achieves exponentially small error probability,  and has quantum depth and total number of qubits
\begin{align}
T=O(1)\ \textrm{ and } \ n^{\mathrm{tot}}=\Theta(n^{3/2}\log^3 n)\ .\label{eq:tonethetan}
\end{align}
Unfortunately, the parameters~\eqref{eq:tonethetan} do not saturate the bound~\eqref{eq:bfsbound} with~$D=3$. 
The fact that these parameters saturate the syndrome extraction circuit depth bound~\eqref{eq:dbtbound} (up to polylogarithmic factors) suggests that this may fundamentally be impossible using constructions that have~$T=O(1)$: Indeed, should an analog of the bound~\eqref{eq:dbtbound} apply without the restriction to parity-controlled Pauli operations, then no LDPC-code based local quantum memory can saturate both bounds. In other words, our construction may be optimal among constant-depth LDPC-code based constructions in terms of qubit overhead.

Let us now consider the hierarchical code. When the paper~\cite{pattisonHierarchicalMemoriesSimulating2023} was posted to arXiv, the single-shot decodability of good quantum LDPC codes had not been established. It is, however, straightforward to apply the subsequent result~\cite{guSingleshotDecodingGood2023} to the hierarchical code construction. This results in a 2D-local quantum memory with exponentially small error probability where the circuit depth and total number of qubits is
\begin{align}
T=\Theta(\sqrt{n}\log n)\ \textrm{ and }\ n^{\mathrm{tot}}=\Theta(n\log^2 n)\ .\label{eq:hierarchicalcodeparameters}
\end{align}
(With the original construction discussed in the paper~\cite{pattisonHierarchicalMemoriesSimulating2023}, the circuit depth was~$T(n)=\Theta(n^{3/2}\log n)$ because 
it involved $d=\Theta(n)$~rounds of (repeated) measurements before decoding the outer LDPC code, following Gottesman's construction~\cite{gottesmanConstantOverhead}). 
The parameters~\eqref{eq:hierarchicalcodeparameters} saturate the bound~\eqref{eq:bfsbound}.

In summary, our construction and the hierarchical code target different regimes, for recovery maps as summarized by the following table:
\begin{table}[!h]
    \centering
    \begin{tabular}{l||l|l|l}
    construction & qubit overhead & quantum depth  & locality\\
    \hline
    \hline
    hierarchical code~\cite{pattisonHierarchicalMemoriesSimulating2023} & polylogarithmic & polynomial & 2D\\
    \hline
    our construction & polynomial & constant & 3D (or quasi-$2D$)
    \end{tabular}
    \caption{Summary on hierarchical codes and quantum memories from our construction.}
    \label{tab:localmemorysummary}
\end{table}

\subsection*{Outline}
We proceed as follows: In Section~\ref{sec:parallelrouting}, we discuss parallel routing of qubits on general graphs and give  algorithms for grid graphs. We explain how to use these to render a quantum circuit local, albeit without fault-tolerance.
In Section~\ref{sec:parallelrepetitiontheorem}, we consider the parallel use of a ``quantum bus'' to concurrently generate several Bell pairs. The main result established there is a parallel repetition theorem that shows robustness against local stochastic noise even in cases where individual buses are affected by correlated errors. In Section~\ref{sec:faulttolerantlylocalizing}, we give our circuit transformation which ``localizes'' a quantum circuit while maintaining fault-tolerance properties. The construction combines our efficient parallel routing schemes with our parallel repetition result for quantum buses. Finally, we describe the protocol of 3D-local fault-tolerant quantum computation with polynomial qubit-overhead and quasi-polylogarithmic time-overhead.

\section{Parallel routing of qubits in 2D and 3D lattices\label{sec:parallelrouting}}
Our construction relies on new schemes for parallel routing of qubits on vertices of a graph. We note that 
the study of routing on an undirected connected graph~$G=(V,E)$ has a long history in classical computer science. 
A prominent example is the routing via matching model introduced by Alon et al.~\cite{AlonChungGraham94}: Here
labeled pebbles are placed on the vertices initially, and are subsequently moved along edges by means of a sequence of steps consisting of (possibly parallel) swaps. The goal is to implement a target permutation of the pebbles using a minimal number of steps; the corresponding number is referred to as the routing time (for a given graph and a permutation). Corresponding routing schemes  find direct application to 
the compilation of circuits which use $\mathsf{SWAP}$-gates to realize long-range gates, see e.g.,~\cite{bapat03}.

\subsection{Parallel routing using entanglement swapping circuits\label{sec:entanglementswappingcircuit}}
Our construction does not use unitary $\mathsf{SWAP}$-gates, but instead relies on entanglement swapping realized by Bell measurements on Bell pairs.  In contrast to  the unitary case, such measurements can create long-range entanglement (e.g., along a path of Bell pairs of length~$n$) in constant time, up to a Pauli correction that can be computed efficiently (in time $O(\log n)$) by taking parities, i.e., classical computation. In particular, this means that with regard to the quantum circuit depth (and assuming  efficient classical computation to be free), the usual figures of merit such as the routing time in the routing via matching problem  are not directly relevant in our scenario. 

The following graph-theoretic notions are motivated by our objective: We are interested in realizing a set of two-qubit interactions (gates) in each time step. Rather than realizing a permutation, this means that we seek to connect pairs of qubits according to the pairwise disjoint supports of a collection of two-qubit operations (e.g., two-qubit gates). The length of any ``path of travel'' is irrelevant (each will be at most of polynomial length) as it does not translate to quantum circuit depth. We only need the paths to be pairwise disjoint. This is because we will assume that there is a Bell pair on each edge  used for entanglement swapping.  Most  importantly, our figure of merit is the number of qubits for which  this kind of routing is possible in a given graph. This leads to the following definitions.

Let $G=(V,E)$ be a connected undirected graph. A subset~$S=\left\{v_1,\ldots,v_{|S|}\right\}\subseteq V$ of vertices is called \emph{parallel-routable} if  its size~$|S|=:2k$ is even and the following holds: For any pairing~$\{(v_{i_r},v_{j_r})\}_{r=1}^k$ of the elements of~$S$, there is a family~$\{\pi_r\}_{r=1}^k$ of $k$~pairwise edge-disjoint paths such that for any $r\in \{1,\ldots,k\}=:[k]$, the path $\pi_r$ 
connects $v_{i_r}$ and $v_{j_r}$. Such a family~$\{\pi_r\}_{r=1}^k$
will be referred to as a parallel routing scheme associated with the pairing~$\{(v_{i_r},v_{j_r})\}_{r=1}^k$. The \emph{parallel routing number} $k(G)=|S|/2$ is defined as  half the maximal size of a subset~$S\subseteq V$ of vertices which is parallel-routable.

For later reference, we make the following obvious observation obtained by restricting a routing scheme:
\begin{lemma}\label{lem:restrictingsubsetroutability}
Let $S\subseteq V$ be a parallel-routable subset of vertices. Let $S'\subset S$ be a subset of even size. Then~$S'$ is parallel-routable.
\end{lemma}

Before studying these combinatorial notions in detail for grid graphs (see Section~\ref{sec:parallelroutingingrid}), let us state how they give rise to routing protocols of qubits on a general graph~$G=(V,E)$. 
We consider two different routing-tasks: Pairwise entanglement generation as well as qubit pairing.

Let $S=\{v_1,\ldots,v_{2k}\}\subset V$ be a parallel-routable set of vertices.
We first discuss pairwise entanglement generation. Here we consider a system with $2k+2|E|$~qubits, namely
\begin{enumerate}[(i)]
\item
A qubit~$P_j$ located at the vertex~$v_{j}$ for each $j\in [2k]$, and 
\item
for each edge~$e=\{u,v\}\in E$, qubits~$R^e_u$, $R^e_v$ located at $u$ and $v$, respectively. 
\end{enumerate} 
Given a pairing~$\{(v_{i_r},v_{j_r})\}_{r\in [k]}$ of the vertices~$S$, the goal of pairwise entanglement generation is 
to generate Bell pairs between each pair $P_{i_r}P_{j_r}$ of qubits, for $r\in [k]$. This is achieved by 
Algorithm~\ref{alg:pairingupqubits}, see
Fig.~\ref{fig:entswappaths} for an illustration.  The following lemma is an immediate consequence of how entanglement swapping works:

\begin{figure}
    \centering
    \includegraphics[width=0.4\textwidth]{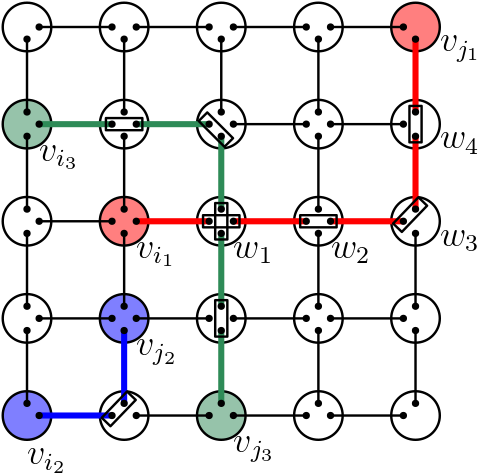}
\caption{Three pairs of parallel-routable vertices~$\{(v_{i_r}, v_{j_r})\}_{r=1}^3$ and their connecting paths. Qubits~$R^e_u$, $R^e_v$ associated with an edge~$e=\{u,v\}$  are represented as black dots. (Qubits $P_j$  located at $v_j$, $j\in [n]$ are not shown.)  Bell measurements in Steps~\ref{it:alongpathone}--\ref{it:italongpathm} are represented by  rectangles. \label{fig:entswappaths}}
\end{figure}
\begin{lemma}
Given a pairing~$(\{(v_{i_r},v_{j_r})\})_{r=1}^{k}$ of the vertices~$S$, the quantum algorithm~$\cQ_{\textrm{entangle}}$ (see Algorithm~\ref{alg:pairwiseentanglement}) 
creates the state
\begin{align}
    \bigotimes_{r=1}^k \Phi_{R_{i_r}R_{j_r}}\ .\label{eq:bellpairspairingent}
\end{align}
It can be realized by a constant-depth adaptive circuit with local operations on the graph~$G$.
\end{lemma}
\begin{proof}
It is clear (from the properties of the basic entanglement swapping circuit) that~$\cQ_{\textrm{entangle}}$ generates the state~\eqref{eq:bellpairspairingent}. 
To see that it can be realized by a constant-depth adaptive circuit, observe that Steps~\ref{it:alongpathone}--\ref{it:italongpathm} in the algorithm can be executed in parallel for different~$r\in [k]$, because the parallel routing scheme~$\{\pi_r\}_{r=1}^k$ consists of pairwise edge-disjoint paths. Similarly, 
Steps~\ref{it:seconditerationx}--\ref{it:seconditerationy} and the steps inside iteration~\ref{it:alongpathxx} can be applied simultaneously for different~$r\in [k]$.
Because the relevant parities in Step~\ref{it:italongpathm} can be computed by 
by classical circuit of depth~$O(\log\ell)$ and $\ell\leq |E|$ is polynomial in~$n$, it follows that $\cQ_{\textrm{pair}}$ can be realized by a constant-depth quantum circuit supplemented with efficient (logarithmic-time) classical computation. 

It is straightforward to verify that all two-qubit gates in this circuit are between neighboring qubits on the graph~$G$, i.e., the circuit is local.
\end{proof}

\begin{algorithm} 
\caption{Quantum algorithm~$\cQ_{\textrm{entangle}}$ for pairwise entanglement generation}\label{alg:pairwiseentanglement}
\begin{flushleft}
  \textbf{Input:} A pairing~$(\{(v_{i_r},v_{j_r})\})_{r=1}^{k}$ of the vertices~$S=\{v_1,\ldots,v_{2k}\}$. \\
    \textbf{Output:}   The  Bell state~$\ket{\Phi}=\frac{1}{\sqrt{2}}(\ket{00}+\ket{11})$ on each pair $R_{i_r}R_{j_r}$ of qubits, $r\in [k]$.\\
  \begin{algorithmic}[1]
  \State Compute a parallel routing scheme~$\{\pi_r\}_{r=1}^k$ associated with the paring~$\{(v_{i_r},v_{j_r})\}_{r=1}^k$. \label{it:stepfirstalgorithmpair}
  \For{$e=\{u,v\} \in E$}
  \State Prepare the state~$\ket{\Phi}=\frac{1}{\sqrt{2}}(\ket{00}+\ket{11})$ on the qubits $R^e_uR^e_v$.
  \EndFor
  \For{$r\in [k]$}\Comment Let $\pi_r=(v_{i_r},w_1,\ldots,w_\ell,v_{j_r})$ be the sequence of vertices traversed along~$\pi_r$.  For a vertex~$v\in \{v_{i_r},w_1,\ldots,w_\ell,v_{j_r}\}$, let  $e_{\textrm{in}}(v)$ and $e_{\textrm{out}}(v)$ denote the in- and outgoing edges at $v$ when traversing~$\pi_r$. 
  \For{$m\in [\ell]$}\label{it:alongpathone}
\State Apply a Bell measurement  to the qubits $R^{e_{\textrm{in}}(w_m)}_{w_m}R^{e_{\textrm{out}}(w_m)}_{w_m}$, getting outcomes~$(a_m,b_m)$.
    \EndFor
   \State     Apply~$X^aZ^b$ to qubit~$R^{e_{\textrm{in}}(v_{j_r})}_{v_{j_r}}$ where $a=\bigoplus_{m=1}^\ell a_m$ and  $b=\bigoplus_{m=1}^\ell b_m$. \label{it:italongpathm}
  \EndFor \label{it:allpathprepared}% for $r\in [k]
    \For{$r\in [k]$}\label{it:alongpathxx}
      \State Apply the $\mathrm{SWAP}$ gate to qubits  $P_{i_r}$ and $R^{e_{\textrm{out}}(v_{i_r})}_{v_{i_r}}$.
      \State Apply the $\mathrm{SWAP}$ gate to qubits  $P_{j_r}$ and $R^{e_{\textrm{in}}(v_{j_r})}_{v_{j_r}}$.
  \EndFor\label{it:alongpathyyentanglement}
    \end{algorithmic}
\end{flushleft}
\end{algorithm}
\noindent

Now consider the problem of pairing up qubits: Here we consider a system of qubits that has -- in addition to the $n+2|E|$ qubits $P_1,\ldots,P_n$ and $\{R^e_u,R^e_v\}_{e=\{u,v\}\in E}$, $n=2k$~``data'' qubits ~$Q_1,\ldots,Q_{2k}$, where qubit $Q_j$ is located at the  vertex~$v_{j}$ for $j\in [2k]$.
Given a pairing~$\{(v_{i_r},v_{j_r})\}_{r\in [k]}$ of the vertices~$S$, the problem of qubit pairing consists in bringing the qubits~$Q_{i_r}$ and $Q_{j_r}$ to the same location, for every~$r\in [k]$.   Algorithm~\ref{alg:pairingupqubits} achieves this by transferring each qubit~$Q_{j_r}$ to $P_{i_r}$ for $r\in [k]$, that is, the qubits~$Q_{i_r}$, $Q_{j_r}$ of the input state are located at~$Q_{i_r}$ and $P_{i_r}$ after application of the algorithm, i.e., at vertex~$v_{i_r}$. The following lemma is an immediate consequence 
of how the standard teleportation protocol works.
\begin{lemma}\label{lem:qubitroutingsimple}
Given a pairing~$\{(v_{i_r},v_{j_r})\}_{r=1}^{k}$ of the vertices~$S$, the quantum algorithm~$\cQ_{\textrm{pair}}$ (see Algorithm~\ref{alg:pairingupqubits}) implements a transfer of subsystems that maps each subsystem~$Q_{j_r}$ to the subsystem $P_{i_r}$ for $r\in [k]$, i.e., a $2k$-qubit state~$\Psi$ on the registers~$Q_1,\ldots,Q_{2k}$ is mapped according to
\begin{align}
\Psi_{Q_{i_1}Q_{j_1}\cdots Q_{i_k}Q_{j_k}} & \mapsto \Psi_{Q_{i_1}P_{i_1}\cdots Q_{i_k}P_{i_k}}\ .\label{eq:actionofcircuitqpair}
\end{align}
The algorithm~$\cQ_{\textrm{pair}}$ can be realized by an adaptive constant-depth quantum circuit with local operations on the graph~$G$ supplemented by efficient classical computation. 
\end{lemma}
\begin{algorithm} 
\caption{A quantum algorithm~$\cQ_{\textrm{pair}}$ for pairing  qubits}\label{alg:pairingupqubits}
\begin{flushleft}
  \textbf{Input:} A pairing~$\{(v_{i_r},v_{j_r})\}_{r=1}^{k}$ of the vertices~$S=\{v_1,\ldots,v_{2k}\}$. \\
  \ \qquad \qquad A $2k$-qubit state~$\Psi_{Q_1\cdots Q_{2k}}$ on the qubits $Q_1,\ldots,Q_{2k}$.\\
    \textbf{Output:}   The permuted state~$\Psi_{Q_{i_1}P_{i_1}\cdots Q_{i_k}P_{i_k}}$\\
  \   \qquad\qquad   where $Q_{i_r}P_{i_r}$ contains the subsystem~$Q_{i_r}Q_{j_r}$ of the original state~$\Psi$ for $r\in [k]$.
\begin{algorithmic}[1]
\State{Run~$\cQ_{\textrm{entangle}}$  with input $\{(v_{i_r},v_{j_r})\}_{r=1}^{k}$ 
to generate $k$ Bell pairs, i.e., the state~$\bigotimes_{r=1}^k \Phi_{P_{i_r}P_{j_r}}$.}\label{it:statentanglementgeneration}
  \For{$r\in [k]$}\label{it:alongpathx}
      \State Apply a Bell measurement  to the qubits  $Q_{j_r}$ and     
      $P_{j_r}$ obtaining outcome~$(a_r,b_r)$.\label{it:seconditerationx}
      \State Apply $X^{a_r}Z^{b_r}$ to $P_{i_r}$. \label{it:seconditerationy}
  \EndFor\label{it:alongpathy}
    \label{it:alongpathyy}
  \end{algorithmic}
\end{flushleft}
\end{algorithm}

\begin{proof}
After the pairwise entanglement generation step~\eqref{it:statentanglementgeneration},
we have $k$ (possibly ``long-range'') Bell states
\begin{align}
\bigotimes_{r=1}^k \Phi_{P_{i_r}P_{j_r}}\ . \label{eq:kfoldtensorproduct}
\end{align}
Steps~\ref{it:alongpathx}--\ref{it:alongpathy} perform teleportation, transferring system~$Q_{j_r}$ to~$P_{i_r}$  for each~$r\in [k]$. It follows immediately that the algorithm~$\cQ_{\textrm{pair}}$ 
has the action given by Eq.~\eqref{eq:actionofcircuitqpair}.

It remains to argue that the algorithm~$\cQ_{\textrm{pair}}$ can be realized by a constant-depth adaptive quantum circuit which is local on the graph~$G$. It is clear by inspection that all operations involved in the algorithm are local, i.e., they act on one or two qubits that are either located at the same vertex~$v\in V$, or at two vertices~$u,v\in V$ connected by an edge~$\{u,v\}\in E$. It is easy to check that~$\cQ_{\textrm{pair}}$ can be realized by an adaptive constant-depth circuit.
\end{proof}
Slightly abusing notation, we denote by~$\cQ_{\textrm{pair}}^{-1}$ the algorithm which, for a given pairing
\begin{align}
    \{(v_{i_r},v_{j_r})\}_{r=1}^k \ , 
\end{align}
implements the inverse of the map~\eqref{eq:actionofcircuitqpair} (i.e., maps $P_{i_r}$ to $Q_{j_r}$ for each~$r\in [k]$). It is clear that~$\cQ_{\textrm{pair}}^{-1}$ can also be realized by an adaptive constant-depth quantum circuit in an analogous fashion
by proceeding creating long-range entanglement, and then teleporting respectively swapping the corresponding subsystems.

In Section~\ref{sec:localization} we explain how to use the routing protocols~$\cQ_{\textrm{pair}}$ and $\cQ_{\textrm{pair}}^{-1}$ in 3D grid graphs in order to realize arbitrary quantum circuits using local operations. We subsequently refine the construction to incorporate fault-tolerance considerations.

\subsection{Parallel routing on 2D and 3D grid graphs\label{sec:parallelroutingingrid}}
In this section, we  study routing schemes in certain graphs: We are interested in architectures associated with (``cuboid'') grid graphs in 2D and 3D. To introduce these, let us first introduce the Cartesian product~$G\times G'$ of two graphs
$G=(V,E)$, $G'=(V',E')$
as follows:  This is the 
graph~$G\times G'$  with  vertex set
\begin{align}
V\times V'=\left\{(v,v')\ |\ v\in V, v'\in V'\right\}\ ,
\end{align}
where $\{(u,u'),(v,v')\}$ is an edge of~$G\times G'$ if and only if
either 
\begin{align}
\left(u=v\quad\textrm{ and }\quad (u',v')\in E'\right)\qquad\textrm{ or }\qquad 
\left(u'=v'\quad\textrm{ and }\quad (u,v)\in E\right)\ .
\end{align}
Let $P_n$ denote the path on $n$~vertices, that is, the graph with vertex set~$[n]$ and edges of the form $\{j,j+1\}$ where $j\in [n-1]$.
 We call~$P_{n_1}\times P_{n_2}$ 
the $n_1\times n_2$ grid graph
and $P_{n_1}\times P_{n_2}\times P_{n_3}$ the $n_1\times n_2\times n_3$ grid graph, where $n_1,n_2,n_3\in\mathbb{N}$. 
These are naturally embedded in $\mathbb{R}^2$ and $\mathbb{R}^3$, respectively, and  the distance~$d(u,v)$ between two verties $u,v$ in these graphs (i.e., the minimal length of a path connecting $u$ and $v$) is simply the $L^1$- or ``Manhattan''-distance.

Consider the $L\times L$ grid graph~$P_L\times P_L$. The following lemma gives a sufficient condition for the existence of a parallel routing scheme associated with a pairing of a subset of vertices of this graph.
\begin{lemma}[Parallel routing in~$P_L\times P_L$]\label{lem:parallelrouting2D}
Let $S=\{v_i=(x_i,y_i)\}_{i=1}^{2k}$ be a subset of vertices of~$P_L\times P_L$ of even size  and let~$\{(v_{i_r},v_{j_r})\}_{r=1}^k$ be an associated pairing.
For brevity, let us write
\begin{align}
\begin{matrix}
(X_r,Y_r) &:=&(x_{i_r},y_{i_r})=v_{i_r}\\
(X'_r,Y'_r)&:=&(x_{j_r},y_{j_r})=v_{j_r}
\end{matrix}\qquad\textrm{ for } r\in [k]\ 
\end{align}
for the coordinates of paired vertices.
Suppose that 
\begin{align}
\begin{matrix}
\{X_{p},X'_{p}\}&\cap&
\{X_{q},X'_{q}\}&=&\emptyset&\qquad\\
\{Y_{p},Y'_{p}\}&\cap&
\{Y_{q},Y'_{q}\}&=&\emptyset&
\end{matrix}
\qquad\textrm{ for all } p,q\in [k]\textrm{ with } p\neq q\ .\label{eq:assumptionnonintersect}
\end{align} 
Then there exists a parallel routing scheme~$\{\pi_r\}_{r=1}^k$ associated with the pairing~$\{(v_{i_r},v_{j_r})\}_{r=1}^k$.
Furthermore, the length of the path~$\pi_r$ is equal to $d(v_{i_r},v_{j_r}) \leq 2L$ for each $r\in [k]$.
\end{lemma}
\begin{proof}
 Without loss of generality, we may assume that   $X_r \leq X_r'$. For each $r \in [k]$, define the path $\pi_r$ connecting $v_{i_r}=(X_r,Y_r)$ and $v_{j_r}=(X'_r,Y'_r)$ 
by the sequence of traversed vertices as 
  \begin{align} \label{eq:pathinthesamefloor}
    \pi_r :=
    \begin{cases}
      \{ (X_r,Y_r), (X_r+1, Y_r), \dots, (X'_r, Y_r), (X'_j, Y_r+1), \dots,  (X'_r, Y_r')\} & \text{if $Y_r \leq Y'_r$} \\
      \{(X_r,Y_r), (X_r+1, Y_r), \dots, (X'_r, Y_r), (X'_r, Y_r-1), \dots,  (X'_r, Y_r')\} & \text{if $Y_r > Y'_r$} \, .
    \end{cases}
  \end{align} The assumption~\eqref{eq:assumptionnonintersect} implies that the paths $\{\pi_r\}_{r=1}^k$ are pairwise edge-disjoint.
  Furthermore, the length of the path~$\pi_r$ is $d(v_{i_r},v_{j_r})$ for each $r\in [k]$ by construction. 
\end{proof}
An immediate consequence of Lemma~\ref{lem:parallelrouting2D} is the following lower bound on the parallel routing number of a  2D~grid graph.
\begin{corollary}\label{cor:2Droutingnumber}
Let $L$ be even. Then $k(P_L\times P_L)\geq L/2$.
\end{corollary}
\begin{proof}
Consider the set~$S=\{v_i:=(i,i)\}_{i=1}^L$ of vertices on the ``diagonal'' of $P_L\times P_L$. 
Let $\{(v_{i_r},v_{j_r})\}_{r=1}^k$ be an arbitrary pairing of~$S$. Then $i_r\neq j_r$ and $\{i_r,j_r\}\cap \{i_s,j_s\}=\emptyset$ for all $r\neq s$, $r,s\in [L]$ by definition of a pairing. It follows immediately that Condition~\eqref{eq:assumptionnonintersect} is satisfied. Since the pairing was arbitrary, this shows that~$S$ is parallel-routable. 
\end{proof}
Next, we consider 3D grid graphs. Our main result is the following.
\begin{theorem}[Parallel routing in $P_L\times P_L\times P_L$] \label{thm:main3Drouting}
Let $L$ be even. Consider the subset 
\begin{align}
[L]^2\times \{1\}&=\left\{
(x,y,1)\ |\ x,y\in [L]
\right\}
\end{align}
of vertices of the grid graph~$P_L\times P_L\times P_{4L}$. Let~$\{(v_{i_r},v_{j_r})\}_{r=1}^{L^2/2}$ be an arbitrary pairing of the vertices belonging to~$[L]^2\times \{1\}$. Then there is a parallel routing scheme~$\{\pi_r\}_{r=1}^{L^2/2}$ associated with~$\{(v_{i_r},v_{j_r})\}_{i=1}^{L^2/2}$ such that each path~$\pi_r$ has length at most~$10L$ for each $r\in [L^2/2]$. 
\end{theorem}
In particular, this immediately implies the following:
\begin{corollary}
Let $L$ be even. Then $k(P_L\times P_L\times P_{4L})\geq L^2/2$. 
\end{corollary}

To prove Theorem~\ref{thm:main3Drouting}   we construct a parallel routing scheme in a greedy way. More precisely, we use that the graph~$P_L\times P_L\times P_{4L}$ contains~$4L$ subgraphs~$P_L\times P_L\times \{Z\}$, $Z\in [4L]$ which we refer to as floors. Each of these subgraphs is isomorphic to~$P_L\times P_L$. Our construction reduces the problem of finding a parallel routing scheme on~$P_L\times P_L\times P_{4L}$ to the 2D case by applying Lemma~\ref{lem:parallelrouting2D} to different floors. This is achieved by Algorithm~\ref{alg:floorassignment}, which computes
the coordinate~$Z_r$ of a floor $P_{L}\times P_L\times \{Z_r\}\subset P_L\times P_L\times P_{4L}$ for each pair~$(v_{i_r},v_{j_r})$,  $r\in [L^2/2]$.   We say that the pair~$(v_{i_r},v_{j_r})$ is assigned the $Z_r$-th floor by Algorithm~\ref{alg:floorassignment}.

\begin{algorithm} 
\caption{A greedy algorithm to assign a floor to each pair of vertices in a grid graph}\label{alg:floorassignment}
\begin{flushleft}
  \textbf{Input:} A pairing~$(\{(v_{i_r},v_{j_r})\})_{r=1}^{L^2/2}$ of the vertices~$
[L]^2\times \{1\}$. \\
\qquad\quad \ We write $v_{i_r} =: (X_r,Y_r,1)$ and $v'_{j_r} =: (X'_r, Y'_r, 1)$ for $r\in [L^2/2]$. \\
  \textbf{Output:} A sequence~$(Z_r)_{r=1}^{L^2/2}$ with $Z_r \in [4L]$.
\begin{algorithmic}[1]
  \State $\mathrm{FilledRows}[1] \leftarrow \emptyset, \dots, \mathrm{FilledRows}[4L] \leftarrow \emptyset$
  \State $\mathrm{FilledCols}[1] \leftarrow \emptyset, \dots, \mathrm{FilledCols}[4L] \leftarrow \emptyset$
  \For{$r \in [L^2/2]$}
%  \For{$Z \in \{2, \dots, 4L\}$}
\State $Z_r\leftarrow \arg\min \left\{Z\in [4L]\ |\ \{X_r,X'_r\}\cap \mathrm{FilledCols}[Z]=\emptyset
\textrm{ and }\{Y_r, Y'_r\}\cap\mathrm{FilledRows}[Z]=\emptyset\right\}$.\label{step:fillcondition}
  \State $\mathrm{FilledCols}[Z_r] \leftarrow \mathrm{FilledCols}[Z_r] \cup \{X_r, X'_r\}$\label{step:filledcols}
  \State $\mathrm{FilledRows}[Z_r] \leftarrow \mathrm{FilledRows}[Z_r] \cup \{Y_r, Y'_r\}$\label{step:filledrows}
  %\State \textbf{break}
 % \EndIf
  %\EndFor
  \EndFor\\
  \Return $(Z_r)_{r=1}^{L^2/2}$
\end{algorithmic}
\end{flushleft}
\end{algorithm}
\noindent The following lemma shows that
the set of pairs assigned to a given floor satisfies the sufficienct
condition~\eqref{eq:assumptionnonintersect}  for the existence of a parallel routing scheme in 2D.
\begin{lemma} \label{lem:greedyzcoordworks}
For any  input pairing~$\{(v_{i_r},v_{j_r})\}_{r=1}^{L^2/2}$ of $
[L]^2\times \{1\}$,   Algorithm~\ref{alg:floorassignment} returns a family~$\{Z_r\}_{r=1}^{L^2/2}\subset [4L]$ with the following property.
  Writing $v_{i_r}=(X_r,Y_r,1)$ and $v_{j_r}=(X'_r,Y_r,1)$ for $r\in [L^2/2]$,   and defining 
  \begin{align}
  I_Z&=\{r\in [L^2/2]\ |\ Z_r=Z\}\ 
  \end{align}as the  set of indices of pairs~$(v_{i_r},v_{j_r})$  assigned to the $Z$-th floor, we have 
  \begin{align}\label{eq:necessaryconditionsatisfied}
  \begin{matrix}
  \{X_p,X'_p\}&\cap& \{X_q,X'_q\}&=&\emptyset\\
    \{Y_p,Y'_p\}&\cap& \{Y_q,Y'_q\}&=&\emptyset
    \end{matrix}\qquad\textrm{for all} \quad  p,q\in I_Z \quad \textrm{with} \quad p\neq q\ 
     \end{align}
  for every $Z\in [4L]$. 
    \end{lemma}
    \begin{proof}
    We first argue 
    that the algorithm is well-defined, i.e., returns a sequence~$\{Z_r\}_{r=1}^{L^2/2}\subset [4L]$
    on any input.     To this end, we have to argue     that for every $r\in [L^2/2]$, the assignment in Step~\ref{step:fillcondition} indeed assigns 
    a value $Z_r\in [4L]$ for every $r\in [L^2/2]$, i.e., there is some $Z\in [4L]$ such that 
$\{X_r,X'_r\}\cap \mathrm{FilledCols}[Z]$ and 
$\{Y_r,Y'_r\}\cap \mathrm{FilledRows}[Z]$ are both empty.
   
   Suppose for sake of contradiction that this is not the case. Let $r\in [L^2/2]$ be 
   the minimal number such that 
  \begin{align}
  \left(\{X_r,X'_r\}\cap \mathrm{FilledCols}[Z]\neq \emptyset\textrm{  or }
  \{Y_r,Y'_r\}\cap \mathrm{FilledRows}[Z]\neq \emptyset \right)\quad\textrm{for every} \ Z\in [4L]\ .\label{eq:violationcondition}
  \end{align} 
  Let us consider the collection of sets~$\{\mathrm{FilledRows}[Z]\}_{Z\in [4L]}$
and~$\{\mathrm{FilledCols}[Z]\}_{Z\in [4L]}$ at the corresponding Step~\ref{step:fillcondition} of the algorithm.   By definition (see Line~\ref{step:filledcols}),  $\mathrm{FilledCols}[Z]$ is the set  of all  $X$-coordinates of pairs $(v_{i_q},v_{j_q})$, $q\leq r-1$ that have been assigned to the $Z$-th floor. Similarly (see Line~\ref{step:filledrows}),  $\mathrm{FilledRows}[Z]$ is the set  of all $Y$-coordinates of pairs $(v_{i_q},v_{j_q})$, $q\leq r-1$  assigned to the $Z$-th floor. Note that in each iteration, at most one pair is assigned to a hitherto  ``unused'' floor (i.e., one associated with a coordinate $Z\in [4L]$ such that~$\mathrm{FilledCols}[Z]$ and $\mathrm{FilledRows}[Z]$ are both empty). Since  Eq.~\eqref{eq:violationcondition} implies that each of the $4L$~floors is already occupied at the $r$-th iteration, we must have 
\begin{align}
r-1&\geq 4L\ .\label{eq:rminusone4L}
\end{align}
The definition of the sets~$\{\mathrm{FilledRows}[Z]\}_{Z\in [4L]}$
and~$\{\mathrm{FilledCols}[Z]\}_{Z\in [4L]}$ and~\eqref{eq:violationcondition}
 also  implies that 
  \begin{align}
 \left(\{X_r,X'_r\}\cap \{X_q,X'_q\}\neq \emptyset\textrm{ or }\{Y_r,Y'_r\}\cap \{Y_q,Y'_q\}\neq \emptyset\right)\qquad\textrm{ for each } q\leq r-1\ .\label{eq:xryrp}
  \end{align}
  Because by~\eqref{eq:xryrp}, at least one of the vertices~$v_{i_q}$ or $v_{j_q}$ belongs to the set
  \begin{align}
  W&:=\left\{(X,Y)\in [L]^2\ |\ 
  X\in \{X_r,X'_r\}\textrm{ or } Y\in \{Y_r,Y'_r\}
  \right\}\ ,
   \end{align}
 the cardinality of this set must be at least
 \begin{align}
 |W| & \geq 4L\ \label{eq:lowerboundvl}
 \end{align}  
 by Eq.~\eqref{eq:rminusone4L}.
 
  However, since $(X_r,Y_r)\neq (X'_r, Y'_r) \in [L]^2$, we have from a simple counting argument that
  \begin{align}
    |W|\leq
    \begin{cases}
      4L - 4\qquad & \text{if $X_r \neq X'_r$ and $Y_r \neq Y'_r$} \\
      3L - 2\qquad & \text{if $X_r = X'_r$ or $Y_r = Y'_r$} \, 
    \end{cases}
  \end{align}
  and thus 
  \begin{align}
  |W| & \leq 4L-4
  \end{align}
  for $L\geq 2$. This contradicts Eq.~\eqref{eq:lowerboundvl}, showing that the algorithm returns a sequence~$\{Z_r\}_{r=1}^{L^2/2}\subset [4L]$ as claimed. 
  
  To see that Condition~\eqref{eq:necessaryconditionsatisfied}
  is satisfied for every $Z\in [4L]$ upon completion of the algorithm,  observe that a corresponding modified  condition is satisfied throughout the algorithm. That is, at any point of the algorithm, a new pair~$((X_q,Y_q),(X'_q,Y_q'))$ only gets 
  assigned to a floor~$P_L\times P_L\times\{Z\}$  if the latter is either empty, or assigning the pair to the floor does not conflict with
the pairs already added to that floor, i.e., does 
  not violate Condition~\eqref{eq:necessaryconditionsatisfied}. This is ensured by the definition of~$Z_r$ in Step~\ref{step:fillcondition}.   
    \end{proof}
        Equipped with 
    Lemma~\ref{lem:greedyzcoordworks}, we can now give the 
\begin{proof}[Proof of Theorem~\ref{thm:main3Drouting}]
For $r\in [L^2/2]$, we define the path~$\pi_r$ connecting the two vertices
\begin{align}
v_{i_r}&=(X_r,Y_r,1)\qquad\textrm{ and }\qquad 
v_{j_r}=(X'_r,Y'_r,1)\ 
\end{align}
as the concatenation of appropriate subpaths~$\pi_r^{\mathrm{up}}$, $\pi_r^{\mathrm{mid}}$ and $\pi_r^{\mathrm{down}}$.
To construct these paths, we use 
 Algorithm~\ref{alg:floorassignment}
to compute -- for each $r\in [L^2/2]$ -- the coordinate
$Z_r\in [4L]$ of a floor $P_L\times P_L\times \{Z_r\}$ 
 of the grid graph~$P_L\times P_L\times P_{4L}$.
The path $\pi_r^{\mathrm{up}}$ ``vertically'' connects~$v_{i_r}$ with the vertex~$(X_r,Y_r,Z_r)$ in the $Z_r$-th floor, and
$\pi_r^{\mathrm{down}}$ vertically connects the vertex~$(X'_r,Y'_r,Z_r)$ in the $Z_r$-th floor with~$v_{j_r}$. 
That is, again defining paths in terms of the vertices traversed, we set
\begin{align}
\begin{matrix}
  \pi_r^{\mathrm{up}} &:=& \{ (X_r, Y_r, 1), (X_r, Y_r, 2), \dots, (X_r, Y_r, Z_r) \} \\
  \pi_r^{\mathrm{down}} &:= &\{ (X'_r, Y'_r, Z_r), (X'_r, Y'_r, Z_r-1), \dots, (X'_r, Y'_r, 1) \} \,.
  \end{matrix}\label{eq:updownpathdefm}
\end{align}
We note that for $Z_r=1$, these paths are empty, i.e., of length~$0$, and can be omitted from further considerations: Here $\pi_r=\pi_r^{\textrm{mid}}$. 

For each $r\in [L^2/2]$, the path~$\pi_r^{\mathrm{mid}}$
connects the vertices~$(X_r,Y_r,Z_r)$ and $(X'_r,Y'_r,Z_r)$ within the $Z_r$-th floor.  To construct these paths, we invoke our construction for the grid graph~$P_L\times P_L$ given in Lemma~\ref{lem:parallelrouting2D}. Lemma~\ref{lem:greedyzcoordworks} guarantees that the collection 
of vertices
\begin{align}
    \left\{(X_r,Y_r,Z_r), (X'_r,Y'_r,Z_r)\right\}_{r\in I_Z}
\end{align}
associated with pairs assigned to the~$Z_r$-th floor satisfies the necessary condition~\eqref{eq:assumptionnonintersect}.  Each path~$\pi_r^{\mathrm{mid}}$ is thus of the form~\eqref{eq:pathinthesamefloor} inside the $Z_r$-th floor, and all path inside a given floor are pairwise non-intersecting.

This constructions satisfies the claimed properties: Denoting by $\abs{\pi}$ the length of a path~$\pi$, we have $\abs{\pi_r^{\mathrm{up}}} \leq 4L$
 and 
$\abs{\pi_r^{\mathrm{down}}} \leq 4L$
by construction, and 
 $\abs{\pi_r^{\mathrm{mid}}} \leq 2L$  by Lemma~\ref{lem:parallelrouting2D}
  and the fact that the diameter of~$P_L\times P_L$ is bounded by~$2L$. It follows that  
  \begin{align}
   \abs{\pi_r}=\abs{\pi_r^{\mathrm{up}}}+\abs{\pi_r^{\mathrm{mid}}}+\abs{\pi_r^{\mathrm{down}}} \leq 10L   
  \end{align}
   for every $r\in [L^2/2]$, as claimed. 

 Furthermore,  the paths $\{\pi_r\}_{r=1}^{L^2/2}$ are pairwise edge-disjoint: Consider two such paths~$\pi_r,\pi_s$ with $r\neq s$. Then it is clear from the definition that the ``vertical'' subpaths $\pi_r^{\mathrm{up}},\pi_r^{\mathrm{down}}$
 and $\pi_s^{\mathrm{up}},\pi_s^{\mathrm{down}}$
 are edge-disjoint,
 as the former are located  above~$v_{i_r}$ and $v_{j_r}$, respectively, whereas the latter are located above $v_{i_s}$ and $v_{j_s}$, respectively.
 Finally, the subpaths~$\pi_r^{\mathrm{mid}}$ and
 $\pi_s^{\mathrm{mid}}$ are edge-disjoint, as they either belong to different floors~$Z_r\neq Z_s$ (and are thus clearly edge-disjoint), or they belong to the same floor and are edge-disjoint by Lemma~\ref{lem:parallelrouting2D} as argued above. 
 
\end{proof}

\subsection{Definition of a general adaptive quantum circuit\label{sec:defadaptivequantumcircuit}}
The following notion will be useful: By a general {\em adaptive quantum circuit~$\cQ$} (sometimes simply referred to as a {\em protocol}) we mean one which is composed single-qubit state preparations, one- and two-qubit unitaries, and single-qubit measurements. 
  All these operations can depend or (i.e., be classically controlled by) measurement results obtained in the course of the computation. These classical controls may involve (efficient) classical computation.

Because of gate teleportation and associated magic states (see~\cite{Magicstate05}), we may assume without loss of generality that~$\cQ$ involves only the following quantum operations on a fixed number~$n$ of qubits: 
\begin{enumerate}[(a)]
\item\label{it:firstoperationqcircuit}
preparation of a single qubit in the computational basis state~$\ket{0}$.
\item
preparation of a single qubit in the magic state~$T\ket{+}=\frac{1}{\sqrt{2}}\left(\ket{0}+e^{i\pi/4}\ket{1}\right)$, where $T = \ketbra{0} + e^{i\pi/4} \ketbra{1}$.
\item
any single- and two-qubit Clifford operation (including the identity, whose inclusion is convenient for formal purposes) 
\item\label{it:lastoperationqcircuit}
measurement of any qubit in the computational basis. This provides a  measurement result~$x\in \{0,1\}
$ and projects the qubit onto the associated $Z$-eigenstate, i.e., the computational basis state~$\ket{x}$.
\end{enumerate}
Adaptivity refers to the fact that the choice of which of these operations is applied at any given point may depend on the history of measurement results obtained up to that point. This dependence should be given in the form of efficiently computable functions. In our presentation, we leave these functions, the classical computation required to evaluate these functions, as well as 
the corresponding classical registers (keeping a record of the measurement results) implicit. In the context of fault-tolerance, we will further assume that they can be evaluated in an error-free manner.

Without loss of generality, we may assume that the algorithm proceeds in an alternating fashion. In each time step, a layer of quantum operations from the set~\eqref{it:firstoperationqcircuit}--\eqref{it:lastoperationqcircuit}
is applied, where the one- respectively two-qubit 
operations within the layer act on 
disjoint subsets of qubits and can be applied simultaneously. This is followed by a classical computation (based on the measurement history) to determine the layer of quantum operations applied at the next time step. We refer to the number~$T$ of layers (time steps) needed to implement~$\cQ$ as the {\em quantum circuit depth}  of~$\cQ$, with the understanding that the classical computation involved is efficient and can be neglected.

Suppressing the dependence on classical measurement results, the algorithm~$\cQ$ can thus be written as a composition
$\cM^{(T)}\circ\cdots\circ\cM^{(1)}$ 
of $n$-qubit operations~$\cM^{(t)}$, $t\in [T]$, where each of these operations consists in the parallel application of one- and two-qubit operations. For simplicity, let us assume that the number~$n=2k$ of qubits is even, and let us label the qubits as $Q_1\cdots Q_n$. Possibly combining single-qubit operations into two-qubit operations,  we can then find, for each~$t\in [T]$, a pairing $\{(i^{(t)}_r,j^{(t)}_r)\}_{r=1}^k$ of $[n]$ (i.e., the set of qubits), and two-qubit operations~$\{\cM^{(t,r)}\}_{r=1}^k$ such that 
\begin{align}
\cM^{(t)}&=\bigotimes_{r=1}^k \cM^{(t,r)}_{Q_{i^{(t)}_r}Q_{j^{(t)}_r}}\label{eq:generallayerm}
\end{align}
is the composition of  the operations~$\cM^{(t,r)}$ applied to the two qubits $Q_{i^{(t)}_r}Q_{j^{(t)}_r}$, which can be executed  in parallel for each~$r\in [k]$. 
\subsection{Localizing a general quantum circuit\label{sec:localization}}
In this section, we show how to 
render a fully general adaptive quantum circuit  local in 3D (respectively in 2D).   That is, let~$\cQ$ be a general adaptive quantum circuit as defined in Section~\ref{sec:defadaptivequantumcircuit}.
 Combining the general qubit pairing protocols~$\cQ_{\textrm{pair}}$ and $\cQ_{\textrm{pair}}^{-1}$ from Section~\ref{sec:entanglementswappingcircuit}
with the routing schemes on grid graphs constructed in Section~\ref{sec:parallelroutingingrid}, we immediately obtain the following.

\begin{lemma}[Localizing a general adaptive circuit]\label{lem:localizingcircuit}
Let $\cQ$ be a general adaptive quantum circuit on $n=2k$~qubits and of quantum depth~$T$ as described. Then there are adaptive quantum circuits~$\cQ'_{\mathrm{2D}}$ and $\cQ'_{\mathrm{3D}}$ with the following properties:
\begin{enumerate}[(i)]
\item
By taking certain marginals (i.e., tracing out qubits and/or ignoring measurement results), the two circuits exactly simulate~$\cQ$.
\item Both circuits have quantum depth of order~$O(T)$.
\item
The circuits $\cQ'_{\mathrm{2D}}$ and $\cQ'_{\mathrm{3D}}$ are geometrically local in 2D and 3D, respectively (i.e., only involve local or nearest-neighbor operations on a corresponding grid graph).
\item
The circuits~$\cQ'_{\mathrm{2D}}$ and~$\cQ'_{\mathrm{3D}}$ use a total number of
\begin{align}
n^{\mathrm{tot}}_{\mathrm{2D}} &= O(n^2) \qquad\textrm{ and }\qquad n^{\mathrm{tot}}_{\mathrm{3D}}= O(n^{3/2}) \label{eq:numberofqubitslocalizedcircuits}
\end{align}
qubits, respectively.
\end{enumerate}
\end{lemma}
\begin{proof}
We use the graphs
\begin{align}
G_{\mathrm{2D}}&=P_n\times P_n\qquad\textrm{ and }\qquad G_{\mathrm{3D}}=P_{\lceil \sqrt{n}\rceil}\times 
P_{\lceil \sqrt{n}\rceil}\times P_{4\lceil \sqrt{n}\rceil}\ ,
\end{align}
respectively. For $G_{\mathrm{2D}}$ we use the parallel-routable subset
\begin{align}
S_{\mathrm{2D}}=\{v_i:=(i,i)\}_{i=1}^n\ ,
\end{align}
see the proof of Corollary~\ref{cor:2Droutingnumber}. For the 3D~grid graph~$G_{\mathrm{3D}}$,
we use an arbitrary subset 
\begin{align}
S_{\mathrm{3D}}&=\{v_i\}_{i=1}^n\subseteq \left[\lceil \sqrt{n}\rceil\right]^2\times \{1\}
\end{align}
of size~$n$. Observe that such a set~$S_{\mathrm{3D}}$ is parallel-routable by Theorem~\ref{thm:main3Drouting} and Lemma~\ref{lem:restrictingsubsetroutability}.

We construct our circuits~$\cQ'_{\mathrm{2D}}$ and $\cQ'_{\mathrm{3D}}$ by replacing, in each layer~$t\in [T]$ of the original adaptive circuit~$\cQ$, the operation~$\cM^{(t)}$ (see Eq.~\eqref{eq:generallayerm}) 
 by a geometrically local adaptive circuit~$\widehat{\cM}^{(t)}$. To do so, we use -- in addition to the 
``computational'' qubits $Q_1,\ldots,Q_{2k}$ --
 a total of
\begin{align}
n^{\textrm{aux}}&=2k+2|E|\label{eq:numberofauxqubits}
\end{align}
 auxiliary qubits with registers labeled~$\{P_j\}_{j\in [2k]}$ and $\{R^e_u,R^e_v\}_{e=\{u,v\}\in E}$. These are arranged as described in Section~\ref{sec:entanglementswappingcircuit}.
We note that the graphs~$G_{\mathrm{2D}}$ and $G_{\mathrm{3D}}$ have
\begin{align}
|E_{\mathrm{2D}}|&=2n^2 - 2n \qquad\textrm{ and } \qquad |E_{\mathrm{3D}}| = 12 \lceil \sqrt{n} \rceil^3 - 9 \lceil \sqrt{n} \rceil^2
\end{align}
edges, respectively, leading (together with~\eqref{eq:numberofauxqubits})
to the claimed total numbers~$n^{\textrm{tot}}=n+n^{\textrm{aux}}$ of qubits, cf.~Eq.~\eqref{eq:numberofqubitslocalizedcircuits}.

Let~$\{(i_r^{(t)},j_r^{(t)})\}_{r=1}^k$ be 
the pairing 
relevant at time step (layer)~$t\in [T]$, and let 
\begin{align}
\cQ^{(t)}_{\textrm{pair}}&=
\cQ_{\textrm{pair}}\left(\{(i_r^{(t)},j_r^{(t)})\}_{r=1}^k\right)\qquad \textrm{ and }\qquad (\cQ^{(t)}_{\textrm{pair}})^{-1}=
\cQ^{-1}_{\textrm{pair}}\left(\{(i_r^{(t)},j_r^{(t)})\}_{r=1}^k\right)\ \label{eq:qpairtpair} 
\end{align}
be the constant-depth adaptive circuits introduced in Section~\ref{sec:entanglementswappingcircuit}.  Then the new circuits~$\cQ_{\mathrm{2D}}'$ respectively~$\cQ_{\mathrm{3D}}'$ are obtained by replacing -- for each~$t\in [T]$ -- the operation
~$\cM^{(t)}$ in the circuit~$\cQ$ 
by the composition
\begin{align}
\widehat{\cM}^{(t)}:=
(\cQ^{(t)}_{\textrm{pair}})^{-1}\circ
\left(\bigotimes_{r=1}^k \cM^{(t,r)}_{Q_{i^{(t)}_r}P_{i^{(t)}_r}}\right)
\circ \cQ^{(t)}_{\textrm{pair}}\ .
\end{align}
In other words, the qubits paired in layer~$t$ are placed next to each other by application of~$\cQ^{(t)}_{\textrm{pair}}$. Then  each 
operation~$\cM^{(t,r)}$, $r\in [k]$ can be applied locally. Subsequently,  the qubits are 
moved back to their original positions by application of $(\cQ^{(t)}_{\textrm{pair}})^{-1}$.  The new circuits then are obtained as the composition
\begin{align}
\cQ'&:=\widehat{\cM}^{(T)}\circ \cdots \circ \widehat{\cM}^{(1)}\ .
\end{align}
It is easy to check that this has all the claimed properties. In particular, for each $t\in [T]$, both  circuits~\eqref{eq:qpairtpair} are adaptive constant-depth circuits, see Lemma~\ref{lem:qubitroutingsimple}. This immediately implies that the quantum circuit depth of~$\cQ'$ is of order~$O(T)$.

 \end{proof}

\section{A parallel repetition theorem for fault-tolerant quantum buses\label{sec:parallelrepetitiontheorem}}
While the (ideal) circuits~$\cQ_{\mathrm{2D}}'$ and~$\cQ_{\mathrm{3D}}'$
realize an adaptive circuit~$\cQ$ with limited overhead in a local manner (see Theorem~\ref{thm:localizingcircuit}), this is clearly not the case for noisy, i.e., imperfect implementations. The most obvious obstacle here concerns the $k$-fold tensor product of Bell pairs~\eqref{eq:kfoldtensorproduct}
associated with a pairing 
(This is generated by the algorithm~$\cQ_{\textrm{entangle}}$ which is used as a subroutine): If
the implementation (of the subroutines~$\cQ_{\textrm{pair}}$ 
respectively~$\cQ_{\textrm{pair}}^{-1}$) is subject to noise, only a noisy version of this central resource state can be generated. Aggravating this  problem, entanglement swapping needs to be performed over extensive (typically polynomial) distances in our construction. Even for i.i.d.~Pauli noise on each qubit, the resulting states when using the ideal circuit~$\cQ_{\mathrm{2D}}'$ or~$\cQ_{\mathrm{3D}}'$ will have exponentially small fidelity with a (tensor product of) Bell state(s) as a function of the distance.

In prior work~\cite{choeLongrangeDataTransmission2022}, we have introduced a fault-tolerant alternative to entanglement-swapping over extensive distances to address this issue. The corresponding protocol  establishes long-range entanglement between two distant qubits even if all involved operations are noisy.
Since this can be combined with teleportation, we refer to the construction as a quantum bus.

Here we consider a more general class of quantum circuits satisfying a certain condition which is also satisfied by the quantum bus.
We show that any quantum circuit consisting of any number of circuits of this class executed in parallel is fault-tolerant. As a result, we have a characterization of the residual ``effective'' noise on the generated tensor product of Bell pairs when several  noisy buses are used in parallel.

In Section~\ref{sec:localstochasticnoise}, we discuss the considered noise model. 
In Section~\ref{sec:noisyimplementation}, we introduce the central notion of a noisy implementation of an adaptive circuit.
In Section~\ref{sec:singlebellstateprep}, we consider a class of fault-tolerant protocol preparing Bell states and check that the quantum bus introduced in Ref.~\cite{choeLongrangeDataTransmission2022} belongs to this class. 
In Section~\ref{sec:ftpreparationparallel}, we establish a parallel repetition theorem characterizing effective noise on the output of the parallel use of quantum circuits introduced in Section~\ref{sec:singlebellstateprep}.
In Section~\ref{sec:busmultiple}, we show as a corollary of the parallel repetition theorem that under any local stochastic noise of strength below a certain threshold, the parallel repetition of the quantum bus generates multiple long-range Bell pairs affected by constant-strength local stochastic noise.

\subsection{Local stochastic noise\label{sec:localstochasticnoise}}
The notion of local stochastic noise was pioneered by Gottesman~\cite{gottesmanConstantOverhead}. It generalizes the restricted noise model of i.i.d.~Pauli noise, and is physically motivated by the idea that physical sources of noise are typically such that errors are  exponentially suppressed in their weight.
Formally, consider a system of $n$ qubits indexed by~$[n]$. 
A random variable $I$ on the power set of $[n]$ is called a local stochastic set of parameter~$p$ if
\begin{align}
    \Pr[F \subseteq I] \leq p^\abs{I} \qquad \text{for all} \qquad F \subseteq [n] \ .
\end{align}
For an $n$-qubit Pauli operator~$E\in\cP_n$, we write~$\supp(E)\subseteq [n]$ for its support, i.e., the set of qubits where~$E$ acts non-trivially. For a random variable~$E$ on the $n$-qubit Pauli group, the support~$\supp(E)$ is a random subset of~$[n]$. The random variable $E$ is called local stochastic noise of strength $p$ if its support~$\supp(E)$ is a local stochastic set of parameter~$p$, i.e.,
\begin{align}
    \Pr[F \subseteq \supp(E)] \leq p^{\abs{F}} \qquad \text{for all} \qquad  F \subseteq [n]  \, .
\end{align}
We write $E \sim \cN(p)$ to denote  local stochastic noise of strength~$p$.

We often consider several errors~$E_1,\ldots,E_r$ each acting on $n$~qubits. In this case, we are dealing with a joint (not necessarily product) distribution~$P_{E_1\cdots E_r}$ 
on~$\cP_n^{\times r}$. In this situation, the expression $E_j\sim \cN(p_j)$ is a statement about the marginal distribution of the $j$-th error~$E_j$.

Local stochastic noise is compatible with Clifford circuits as expressed by simple computational rules, see e.g.,~\cite[Lemma 11]{bravyiQuantumAdvantageNoisy2020} or~\cite{fawziConstantOverheadQuantum2018}.
These rules allow to ``commute'' local stochastic errors forwards or backwards in time through such a circuit. We give quantitative statements in Appendix~\ref{sec:lemmasnoisycircuit}  (see Lemma~\ref{lem:propertieslocalstochasticnoise}). There we also discuss consequences e.g., for entanglement swapping and teleportation under noise.

\subsection{Noisy implementations of general adaptive circuits\label{sec:noisyimplementation}}
Here we formally state the notion of noisy implementation of an adaptive quantum circuit. Recall from Section~\ref{sec:defadaptivequantumcircuit}
that a  general adaptive quantum circuit~$\cQ$ on $n$ qubits can be written as a composition~
\begin{align}
    \cM^{(T)} \circ \dots \circ \cM^{(1)}
\end{align} where each $n$-qubit operation~$\cM^{(t)}$ is composed of the parallel application of one- and two-qubit operations. 
Let $\{E^{(t)}\}_{t=0}^T$ be a $(T+1)$-tuple of $n$-qubit Pauli errors. A noisy implementation of~$\cQ$ with errors~$\{E^{(t)}\}_{t=0}^T$ is then defined as the composition
\begin{align}
    \label{eq:noisyimplementation}
    \cU_{E^{(T)}}\circ\cM^{(T)} \circ  \dots \circ \cU_{E^{(1)}}\circ\cM^{(1)} \circ \cU_{E^{(0)}}
\end{align}
where $\cU^{(t)}_{E^{(t)}}(\rho)=E^{(t)}\rho (E^{(t)})^\dagger$ is  unitary evolution associated with~$E^{(t)}$ for all $t \in [T]$.

Eq.~\eqref{eq:noisyimplementation} may be interpreted as follows: The error~$E^{(0)}$
could be the result of an imperfect preparation procedure for preparing the initial state~$\ket{0^n}$. The error~$E^{(t)}$ describes the effect of a noisy implementation of~$\cM^{(t)}$
(although some of this effect may also be captured by the error~$E^{(t-1)}$). 
Recall that~$\cM$ consists of one- and two-qubit operations applied simultaneously. The role of the error depends on the operation: When preparing a single qubit state~$\ket{0}$
or $T\ket{+}$, the corresponding part (factor) of the Pauli error~$E^{(t)}$ corresponds to the preparation of a corrupted version. For single- or two-qubit Cliffords or the identity operation (``wait locations''), the error amounts to an incorrect implementation of the corresponding unitary. Finally, for  a computational basis measurement, the error corresponds to an erroneous bit-flip of the measurement result: Here only Pauli-$X$-errors are relevant. In summary, errors in a noisy implementation of the circuit~$\cQ$ 
describe fault processes related to all aspects involved in its execution, including preparation, unitary evolution, storage and measurements.

A {\em noisy implementation of $\cQ$ 
with local stochastic noise of strength~$p$}
now refers to  
 a process of the form~\eqref{eq:noisyimplementation}, 
 where  $E^{(t)}\sim\cN(p)$ is local stochastic noise of strength~$p$.
In other words, it is obtained by randomly drawing an~$(T+1)$-tuple~$\{E^{(t)}\}_{t=1}^T$ of $n$-qubit Pauli errors according to a joint distribution (having constraints on the marginals), inserting these errors between gate layers, i.e., running the corresponding noisy implementation of~$\cQ$ with errors~$\{E^{(t)}\}_{t=1}^T$. As mentioned above, the random variables~$\{E^{(t)}\}_{t=1}^T$ may not be independent.

\subsection{Fault-tolerant preparation of a Bell state\label{sec:singlebellstateprep}}
For concreteness, let us give the argument for state preparation protocols, which is the case of interest for our purposes. More specifically, let $\pi$ be an adaptive quantum circuit on~$N$ qubits  which prepares a target (stabilizer) state~$\Phi\in (\mathbb{C}^2)^{\otimes r}$ as follows:
\begin{enumerate}
\item\label{it:stepclifford}
First,  a constant-depth Clifford circuit~$W$ is applied to $\ket{0^N}$. This results in a short-range entangled resource state~$\Psi\in (\mathbb{C}^2)^{\otimes N}$.
\item\label{it:stepmeasurement}
Subsequently, Pauli measurements (i.e., von Neumann measurements associated with 
the eigenbases of different single-qubit Pauli operators) are applied to $N-r$~qubits, resulting in a measurement results~$z\in \{0,1\}^{N-r}$.
\item\label{it:steppostprocess}
A Pauli correction~$C(z)\in \cP_r$ depending (deterministically) on the measurement results~$z$ (i.e., a classically controlled Pauli) is applied to the remaining $r$~qubits. We require that~$C:\{0,1\}^{N-r} \to \cP_r$ has the following properties: First, there is an efficient (classical) algorithm computing~$C(z)$ from~$z$. Second, if Steps~\ref{it:stepmeasurement} and~\ref{it:steppostprocess} are applied to the state~$E \Psi$ where~$E \in \cP_N$, then~$C(z)$ is determined by~$E$.
\end{enumerate}

Now suppose we apply Steps~\ref{it:stepmeasurement} and~\ref{it:steppostprocess}  to the state~$E\Psi$.
 It is straightforward to check that the output of the circuit then is a Pauli-corrupted version~$E_{\textrm{eff}}(E)\Phi$
of the target state~$\Phi$, where the ``effective error''  $E_{\textrm{eff}}(E)$ depends deterministically on the error~$E$.
We are interested in errors~$E$ that lead to a non-trivial overall effective error~$E_{\textrm{eff}}$. To this end, let us define 
\newcommand*{\failset}{\mathsf{FAIL}}
\begin{align}\label{eq:defoffailset}
\failset:=\left\{E\in \cP_N\ |\ E_{\textrm{eff}}(E)\Phi\not\propto \Phi\right\} 
\end{align}
as the set of such errors, i.e., errors~$E$ for which neither~$E_{\textrm{eff}}(E)$ nor $-E_{\textrm{eff}}(E)$ belongs to the stabilizer group of~$\Phi$. If $E$ is a stochastic Pauli error, then the probability that the procedure~$\pi$ does not prepare the desired target state~$\Phi$ is equal to~$\Pr\left[\pi\textrm{ fails}\right]=\Pr[E\in\failset]$.

Robustness of~$\pi$ against local stochastic noise~$E\sim\cN(p)$ follows from certain combinatorial features of the set~$\failset$. It is useful to formalize this as follows. A circuit~$\pi$ will be called $(p_0,f)$-robust  for a threshold error strength~$p_0>0$ and a function $f:[0,p_0]\rightarrow [0,1]$ if there is a family~$\{D_m\}_{m\in\cM}\subset 2^{[N]}$ of subsets of~$[N]$ such that the following two conditions are satisfied:
\begin{enumerate}[(a)]
\item\label{it:decompositionpropertyone}
For every $E \in \failset$ there exists an element $m \in \cM$ such that $D_m \subseteq \supp(E)$.
\item\label{it:decompositionpropertytwo}
We have 
\begin{align}
\sum_{m\in\cM}p^{|D_m|} &\leq f(p)\qquad\textrm{ for all }\qquad p\leq p_0\ .
\end{align}
\end{enumerate}
The significance of this definition is the following immediate consequence of the union bound.
\begin{lemma}\label{lem:singlefaulttolerantcircuit}
A $(p_0,f)$-robust circuit~$\pi$ fault-tolerantly prepares the state~$\Phi$ 
in the presence of local stochastic noise~$E\sim\cN(p)$ in the sense that
\begin{align}
\Pr\left[\pi\textrm{ fails}\right]&\leq f(p)\qquad\textrm{ for any }\qquad p\leq p_0\ .
\end{align}
\end{lemma}
\begin{proof}
Consider local stochastic noise~$E\sim\cN(p)$ with strength~$p\leq p_0$. Then 
 we have
\begin{align}
\Pr\left[E\in\failset \right]&\leq \sum_{m\in\cM} \Pr\left[D_m \subseteq \supp(E)\right]\qquad\textrm{ by the union bound and~\eqref{it:decompositionpropertyone}}\\
&\leq \sum_{m\in\cM} p^{|D_m|}\qquad\textrm{ by the assumption }E\sim\cN(p)\\
&\leq f(p)\qquad\textrm{ by~\eqref{it:decompositionpropertytwo}}\ .
\end{align}
\end{proof}
In Ref.~\cite{choeLongrangeDataTransmission2022}, 
we followed the proof strategy outlined here to show that a certain circuit~$\pi$ fault-tolerantly prepares long-range entanglement, that is, a Bell state between two distant qubits. Following the seminal work~\cite{RaussendorfBravyiHarrington05}, our construction relies on the use of a 3D~cluster state~$\Pi$ on a certain cubic lattice~$\cL_{\textrm{cluster}}\subset\mathbb{R}^3$
 specified by its linear dimensions~$d\times d\times R$ (here the Bell state is between qubits at opposite ends, i.e., at distance~$R$).  For ease of presentation, we will use the grid graph~$P_d\times P_d\times P_R$
in the following and identify 
certain sites of the cluster state lattice~$\cL_{\textrm{cluster}}$ with vertices of this grid graph. That is, we assume that there is a qubit at each vertex of~$P_d\times P_d\times P_R$.
Then we can express the result of~\cite{choeLongrangeDataTransmission2022} as follows.

\begin{corollary}[Quantum bus architecture]\label{cor:quantumbus}
For any integer~$R\geq 2$ and any $\Delta\in\mathbb{N}$ satisfying
\begin{align}
\Delta\geq 8\log R\ ,\label{eq:deltainequalitycondition}
\end{align}
there is a circuit~$\pi(\Delta,R)$  with the following properties:
\begin{enumerate}[(i)]
\item 
The circuit~$\pi(\Delta,R)$ acts on $N=\Delta^2R$ qubits arranged on the grid graph~$P_\Delta\times P_\Delta\times P_R$, with one qubit placed on each vertex.
\item 
In the absence of errors, the circuit~$\pi(\Delta,R)$ prepares the two-qubit Bell state~$\ket{\Phi}=\frac{1}{\sqrt{2}}(\ket{00}+\ket{11})$ on the qubits~$Q_1Q_2$, where $Q_1$ and $Q_2$ located at the sites~$(1,1,1)$ and $(1,1,R)$, respectively. 
\item The circuit $\pi(\Delta,R)$ follows Steps~\ref{it:stepclifford}--\ref{it:steppostprocess} in Section~\ref{sec:singlebellstateprep}: It
\begin{enumerate}[(1)]
\item first prepares a short-range resource state~$\Psi\in (\mathbb{C}^2)^{\otimes N}$ by applying a Clifford circuit~$W$ to the product state~$\ket{0^N}$. The circuit~$W$ is geometrically local on the graph~$P_\Delta\times P_\Delta\times P_R$ and its depth is upper bounded by~$10$.
\item
 then performs single-qubit Pauli measurements on all qubits except~$Q_1Q_2$, obtaining an outcome~$z\in \{0,1\}^{N-2}$.
\item
 finally  applies a two-qubit correction~$C(z)=Z^{\alpha(z)}X^{\beta(z)}$ to qubit~$Q_2$, where $\alpha(z),\beta(z) \in \{0,1\}$ are computed by efficient (classical) algorithms from $z$.
\end{enumerate}
\item 
The circuit $\pi(\Delta,R)$ is  $(p_0,f(p)=p/p_0)$-robust  where $p_0:=1/5004$.
\end{enumerate}
\end{corollary}

Corollary~\ref{cor:quantumbus} is  a direct consequence of Ref.~\cite[Theorem 7.3]{choeLongrangeDataTransmission2022}. The latter result gives a fault-tolerant scheme for preparing a Bell pair on two qubits~$Q_1,Q_2$ at the opposite ends (i.e., distance~$R$) of a certain lattice~$\cL_{\textrm{cluster}}(d,R)$ of linear dimensions of the form $O(d)\times O(d)\times R$ specified by certain pairs of parameters~$(d,R)$ (Here~$d$ is the distance of an underlying surface code.) 
The protocol proceeds as described in Section~\ref{sec:singlebellstateprep}, that is:
\begin{enumerate}
\item 
It starts with an $N=\Theta(d^2R)$-qubit  state~$\ket{0^N}$ where each qubit is associated with a vertex of~$\cL_{\textrm{cluster}}$.   Then a depth-$9$ Clifford circuit~$W=W(d,R)$ is applied to prepare a resource state~$\Psi$. The state~$\Psi$ is the cluster state on~$\cL_{\textrm{cluster}}$ up to local unitaries.
\item
Every qubit except the two qubits~$Q_1$ and~$Q_2$
are measured in the computational basis, yielding a measurement result~$z \in \{0,1\}^{N-2}$. 
\item 
A Pauli correction $Z^{\alpha(z)}X^{\beta(z)}$ is applied to qubit~$Q_2$,
where  the functions
\begin{align}
    \alpha,\beta:\{0,1\}^{N-2} \rightarrow \{0,1\}  
\end{align}
are efficiently computable. 
\end{enumerate}
Let us call this protocol~$\pi_{\textrm{cluster}}(d,R)$. 
The proof of~\cite[Theorem 2.1]{choeLongrangeDataTransmission2022} implies the following:
\begin{claim}\label{claim:reformulatedbus}
 Suppose that $R\geq 3$ is odd and 
\begin{align}
4\log R\leq d\ .\label{eq:logRdrelationcondition}
\end{align}
Then the protocol~$\pi_{\textrm{cluster}}(d,R)$ is $(p_0,f(p)=p/p_0)$-robust with $p_0:=1/5004$.
\end{claim}
\begin{proof}
The proof of~\cite[Theorem 2.1]{choeLongrangeDataTransmission2022}
shows that the protocol $\pi_{\textrm{cluster}}(d,R)$ is $(p_0,f(p)=p/p_0)$-robust with $p_0:=\frac{1}{5004}$ and $f(p) := 5004p$
under the  assumption that $R\geq 3$ is odd, $d\geq 3$ and 
    \begin{align}
    R &\leq \frac{1}{d(10\sqrt{p_0})^{d-2}}\ .\label{eq:inequalitynecessaryconditionbus}
\end{align}
Eq.~\eqref{eq:inequalitynecessaryconditionbus} can be written as
\begin{align}
\log d+\log R&\leq (d-2)\log \frac{1}{10\sqrt{p_0}}=:(d-2)\alpha\ .
\end{align}
With the inequality $d/4\leq (d-2)\alpha-\log d$ (valid for $d\geq 3$) we conclude that Eq.~\eqref{eq:inequalitynecessaryconditionbus} is satisfied if
\begin{align}
4\log R&\leq d\ ,
\end{align}
for $d\geq 3$. Since we assumed $R\geq 3$, and thus $4\log R>4$, we can omit the requirement $d\geq 3$.  This shows that the circuit~$\pi_{\textrm{cluster}}(d,R)$ is $(p_0,f(p)=p/p_0)$-robust for any odd~$R\geq 3$ and $d\geq 4\log R$.
\end{proof}

 Let us now argue that Claim~\ref{claim:reformulatedbus} implies Corollary~\ref{cor:quantumbus}.

\begin{proof}[Proof of Corollary~\ref{cor:quantumbus}]
Suppose $(R,\Delta)$ satisfying~\eqref{eq:deltainequalitycondition} are given. We separately analyze the following cases:
\begin{description}
\item[Suppose $R=2$.]
We give a circuit~$\pi(\Delta,R)$ which is $(1,f(p)=2p)$-robust, with a unitary~$W$ of depth~$2$.
The circuit~$\pi(\Delta,R)$ is non-adaptive, and simply consists of the Clifford unitary~$W=\mathsf{CNOT}(H\otimes I)$ applied to the two qubits~$Q_1,Q_2$. Under local stochastic noise~$E\sim \cN(p)$ applied after the application of~$W$, the resulting state~$E\ket{\Phi}$ is Pauli-corrupted if  $E_{Q_2}E_{Q_1}^T$ is  non-trivial. We can upper bound this by the probability that either qubit experiences an error, i.e.,
\begin{align}
\Pr\left[\pi\textrm{ fails}\right] &\leq \Pr\left[\mathsf{supp}(E)\cap \{Q_1,Q_2\}\neq \emptyset\right]\\
&\leq \Pr\left[Q_1\in \mathsf{supp}(E)\right]+ \Pr\left[Q_2\in \mathsf{supp}(E)\right]\\
&\leq 2p
\end{align}
where we used the union bound and local stochasticity.
The claim follows from this. 

\item[Suppose~$R\geq 3$ is odd.]
We show that there is a protocol
$\pi(\Delta,R)$ which is $(1/5004,f(p)=5004p)$-robust, with a unitary~$W$ of depth~$9$.  Set
$d:=\big\lfloor \frac{\Delta+1}{2}\big\rfloor$. 
Then $2d-1\leq \Delta$, and thus the grid graph~$P_{2d-1}\times P_{2d-1}\times P_R$ is a subgraph of~$P_{\Delta}\times P_{\Delta}\times P_R$. Let $\pi(\Delta,R)$ be the protocol obtained by applying~$\pi_{\textrm{cluster}}(d,R)$ to the qubits belonging to this subgraph. 
To check that~$\pi(\Delta,R)$ is $(1/5004,f(p)=5004p)$-robust, it suffices  to check that condition~\eqref{eq:logRdrelationcondition}
for Claim~\ref{claim:reformulatedbus} 
  is satisfied. But this follows immediately from the inequality~$\frac{\Delta}{2}\leq \lfloor\frac{\Delta+1}{2}\rfloor$ and assumption~\eqref{eq:deltainequalitycondition}.  % NEED $\Delta\geq 8\log R$

\item[Suppose~$R\geq 4$ is even.]
We give a protocol~$\pi(\Delta,R)$ which is $(1/5004,f(p)=5004p)$-robust, with a unitary~$W$ of depth~$10$. 
Suppose $(R,\Delta)$ satisfying~\eqref{eq:deltainequalitycondition} are given with $R\geq 4$ even.
Set $R':=R-1$. Then $P_\Delta\times P_\Delta\times P_{R'}$ is a subgraph of~$P_\Delta\times P_\Delta\times P_R$. To define the circuit~$\pi(\Delta,R)$ in this case, consider the protocol~$\pi(\Delta,R')$ which could be applied to this subgraph (since~$R'\geq 3$ is odd by definition). 
Let $W(\Delta,R')$ be the corresponding unitary, and let $Q_1=(1,1,1)$ and $Q_2'=(1,1,R')$.
Applying~$\pi(\Delta,R')$ would generate entanglement between $Q_1$ and $Q_2'$, but our objective is to entangle the qubit~$Q_1$ with the qubit~$Q_2=(1,1,R'+1)=(1,1,R)$. It is easy to check that this can be achieved by a slight modification, leading to the following definition of~$\pi(\Delta,R)$:
\begin{enumerate}
\item
Set~$W(\Delta,R):=\mathsf{SWAP}_{Q_2Q_2'}W(\Delta,R')$, i.e., succeed the application of the Clifford unitary~$W(\Delta,R')$ with a $\mathsf{SWAP}$-gate between $Q_2$ and $Q_2'$. 
\item
Perform all single-qubit measurements and obtain measurement results~$z$ (as in the protocol~$\pi(\Delta,R')$).
\item
Apply the correction~$C(z)=Z^{\alpha(z)}X^{\beta(z)}$ to qubit~$Q_2$ (instead of qubit~$Q'_2$).
\end{enumerate}
Condition~\eqref{eq:logRdrelationcondition} for $R'$ and $d$ is again satified here because 
\begin{align}
4\log R'&\leq 4\log R\qquad\textrm{ since }R=R'+1\\
&\leq  \frac{\Delta}{2}\qquad\textrm{ by the assumption~\eqref{eq:deltainequalitycondition}}\\
&\leq \left\lfloor\frac{\Delta+1}{2}\right\rfloor=d\ .
\end{align}
By Claim~\ref{claim:reformulatedbus}, this shows that the protocol~$\pi(\Delta,R)$ is $(1/5004,f(p)=5004p)$-robust. 
\end{description}
\end{proof}

\subsection{Fault-tolerant preparation of several Bell states in parallel\label{sec:ftpreparationparallel}}

The result of Lemma~\ref{lem:singlefaulttolerantcircuit} can be generalized to the case where  multiple quantum circuits are run in parallel. This is expressed by the following theorem:
\begin{theorem}[Parallel repetition]\label{thm:parallelrepetition}
For $j\in [k]$, let $\pi^{(j)}$ be an adaptive quantum circuit on $N^{(j)}$~qubits which prepares a target state~$\Phi^{(j)}\in (\mathbb{C}^2)^{\otimes r^{(j)}}$, and is $(p_0^{(j)},f^{(j)})$-robust.
Consider the circuit~$\pi=\pi^{(1)}\times\cdots \times \pi^{(k)}$ obtained by 
running each of these circuits in parallel on $N=\sum_{j=1}^{k}N^{(j)}$ qubits.
Then there are constants~$C>0$, $c\in (0,1)$ such that the following holds. 
Let $E\sim\cN(p)$ be local stochastic noise on~$N$ qubits with strength~$p\leq \min\{p_0^{(j)}\}_{j=1}^k =: p_0$. 
Then  the output state of~$\pi$ is 
$E_{\mathrm{eff}}\left(\bigotimes_{j=1}^k \Phi^{(j)}\right)$ where $E_{\mathrm{eff}}$ is local stochastic noise on $\sum_{j=1}^{k} r^{(j)}$ qubits of strength~
\begin{align}\label{eq:strengthofEffgeneral}
E_{\mathrm{eff}}\sim \cN(f(p)) \qquad \text{with} \qquad f(p) = \left(\max_{j \in [k]} f^{(j)}(p)\right)^{\frac{1}{r}} \ , \quad r = \max_{j \in [k]} r^{(j)} \ .
\end{align}
In particular, if there are constants $C > 0$ and $c \in (0,1)$ such that 
\begin{align}
    f^{(j)}(p) \leq C p^c\qquad\textrm{ for all }p \in [0, p_0]\textrm{ and }j \in [k]\ ,\label{eq:uniformupperboundnoisestrength}
\end{align} then
\begin{align}\label{eq:strengthofEffspecific}
    E_{\mathrm{eff}} \sim \cN(C^{1/r}p^{c/r}) \ 
\end{align}
whenever $p\leq p_0$. 
\end{theorem}

\begin{proof}
We assume here that we apply the measurement- and correction-parts of each circuit~$\pi^{(j)}$ to a state of the form
\begin{align}
E\left(\bigotimes_{j=1}^k\Psi^{(j)}_{\cC^{(j)}}\right)\ ,\label{eq:clusterstatecorrupted}
\end{align}
with local stochastic noise~$E\sim\cN(p)$ acting on all qubits.
 We can factor~$E$ as 
\begin{align}
E&=\bigotimes_{j=1}^k E^{(j)}\ ,\label{eq:multibuserror}
\end{align}
where now~$E^{(j)}$ acts only on the qubits associated with~$\cC^{(j)}$, the set of $N^{(j)}$ qubits that circuit~$\pi^{(j)}$ is applied to.  Recall from Lemma~\ref{lem:propertieslocalstochasticnoise}~\eqref{it:subset} that $E^{(j)}\sim\cN(p)$ is local stochastic noise with the same parameter as~$E$, a fact we will use below.

Let $E_{\textrm{eff}}$ be the effective error 
on the~$\sum_{j=1}^k r^{(j)}$ ``target'' 
qubits~$\bigcup_{j=1}^k\{Q^{(j)}_{1},\ldots,Q^{(j)}_{r^{(j)}}\}$. Let 
\begin{align}
\cL\subseteq \bigcup_{j=1}^k\{Q^{(j)}_{1},\ldots,Q^{(j)}_{r^{(j)}}\}
\end{align}
be a subset of the target qubits. We need to upper bound~$\Pr\left[\cL\subseteq \supp(E_{\textrm{eff}})\right]$.

 We can again factor
\begin{align}
E_{\textrm{eff}}&=\bigotimes_{j=1}^k E^{(j)}_{\textrm{eff}}\ ,
\end{align}
where each factor $E^{(j)}_{\textrm{eff}}$ is an $r^{(j)}$-qubit Pauli operator which depends deterministically on~$E^{(j)}$. 
The event~$\cL\subseteq \supp(E_{\textrm{eff}})$ means that 
for at least~$\ell:=\lceil |\cL|/(\max_{j\in [k]}r^{(j)})\rceil$ of the instances, the desired output is non-trivially corrupted by a Pauli error. Let us assume without loss of generality (by reindexing if necessary) that
this affects the first $\ell$ output states~$\Phi^{(1)},\ldots,\Phi^{(\ell)}$. In other words,
all the corresponding circuits failed. Thus we obtain the upper bound
\begin{align}
\Pr\left[\cL\subseteq \supp(E_{\textrm{eff}})\right]&\leq \Pr\left[\pi^{(j)}\textrm{ fails for every }j\in [\ell]\right]\ .
\end{align}

Recall that since~$\pi^{(j)}$ is $(p_0^{(j)}, f^{(j)})$-robust for each $j \in [k]$, there is a family~$\{D^{(j)}_m\}_{m \in \cM^{(j)}}$ of subsets of~$\supp(E^{(j)})$ satisfying conditions~\eqref{it:decompositionpropertyone} and~\eqref{it:decompositionpropertytwo} in Section~\ref{sec:singlebellstateprep}.
By condition~\eqref{it:decompositionpropertyone}, we have that
the event that $\pi^{(j)}$ fails, which is the event that $E^{(j)}\in\failset^{(j)}$, implies the existence of $m^{(j)}\in\cM^{(j)}$ such that $D^{(j)}_{m^{(j)}}\subseteq \supp(E^{(j)})$. 
Thus we obtain (again using the union bound) that
\begin{align}
\Pr\left[\cL\subseteq \supp(E_{\textrm{eff}})\right]&\leq 
\sum_{m^{(1)}\in \cM^{(1)},\ldots,m^{(\ell)}\in \cM^{(\ell)}} \Pr\left[D^{(j)}_{m^{(j)}}\subseteq 
\supp(E^{(j)})\textrm{ for all }j\in [\ell]\right]\label{eq:intermedupabd} \ .
\end{align}
But the condition
\begin{align}
D^{(j)}_{m^{(j)}}\subseteq \supp(E^{(j)})\textrm{ for all }j\in [\ell]
\end{align}
is equivalent to
\begin{align}
\bigcup_{j=1}^{\ell}D^{(j)}_{m^{(j)}}\subseteq \bigcup_{j=1}^{\ell }\supp(E^{(j)})=\supp(E)\ \label{eq:reformulatedunioncondition}
\end{align}
since the sets of qubits~$\{D^{(j)}_{m^{(j)}}\}_{j=1}^k$ (respectively $\{\supp(E^{(j)})\}_{j=1}^k$) are pairwise disjoint (i.e., belong to different instances).
For the same reason, we have 
\begin{align}
\left|\bigcup_{j=1}^{\ell}D^{(j)}_{m^{(j)}}\right|&=\sum_{j=1}^\ell |D^{(j)}_{m^{(j)}}|\ .\label{eq:disjointunionsize}
\end{align}
Inserting~\eqref{eq:reformulatedunioncondition} 
into~\eqref{eq:intermedupabd} gives
\begin{align}
\Pr\left[\cL\subseteq \supp(E_{\textrm{eff}})\right]&\leq 
\sum_{m^{(1)}\in \cM^{(1)},\ldots,m^{(\ell)}\in \cM^{(\ell)}} \Pr\left[
\bigcup_{j=1}^{\ell}D^{(j)}_{m^{(j)}}\subseteq\supp(E)
\right]\\
&\leq 
\sum_{m^{(1)}\in \cM^{(1)},\ldots,m^{(\ell)}\in \cM^{(\ell)}}
p^{\sum_{j=1}^\ell |D^{(j)}_{m^{(j)}}|}\\
&=\prod_{j=1}^\ell \left(\sum_{m^{(j)}\in\cM^{(j)}} p^{|D^{(j)}_{m^{(j)}}|}\right)
\end{align}
where we used~\eqref{eq:disjointunionsize}
and the fact that $E\sim\cN(p)$ is local stochastic.
By condition~\eqref{it:decompositionpropertytwo}, we obtain the upper bound
\begin{align} \label{eq:boundwithellproduct}
\Pr\left[\cL\subseteq \supp(E_{\textrm{eff}})\right]&\leq 
\prod_{j=1}^\ell f^{(j)}(p)
\end{align}
for all $p\leq \min\{p^{(1)}_0,\ldots,p^{(\ell)}_0\}$. Inserting  the definition of~$\ell$ into~\eqref{eq:boundwithellproduct}, we have
\begin{align}\label{eq:boundwithfjandrj}
    \Pr\left[\cL\subseteq \supp(E_{\textrm{eff}})\right]&\leq 
    \left( \max_{j \in [\ell]} f^{(j)}(p) \right)^{\Big\lceil \frac{\abs{\cL}}{\max_{j' \in [\ell]} r^{(j')}} \Big\rceil} \leq \left( \max_{j \in [\ell]} f^{(j)}(p) \right)^{\frac{\abs{\cL}}{\max_{j' \in [\ell]} r^{(j')}} } \ .
\end{align}
The expression~\eqref{eq:strengthofEffgeneral} for noise strength of~$E_{\mathrm{eff}}$ follows from
the upper bound~\eqref{eq:boundwithfjandrj} combined with the bounds~$\max_{j \in [\ell]} f^{(j)} \leq \max_{j \in [k]} f^{(j)}$ and $\max_{j' \in [\ell]} r^{(j')} \leq \max_{j' \in [k]} r^{(j')}$ for $\ell \leq k$ together with the definition of local stochastic noise. 
Assuming that~\eqref{eq:uniformupperboundnoisestrength}, the expression~\eqref{eq:strengthofEffgeneral}
for the noise strength immediately follows from Eq.~\eqref{eq:strengthofEffspecific}.
\end{proof}
We note that the noise strength of an effective error in Theorem~\ref{thm:parallelrepetition} can be improved if we know about the stabilizer groups of~$\Phi^{(j)}$. Recall that the effective error~$E_{\mathrm{eff}}^{(j)}(E)$ can be chosen up to stabilizers of the $j$-th target state~$\Phi^{(j)}$. Let us choose $E_{\mathrm{eff}}^{(j)}(E)$ to be of minimal support, and define
\begin{align}
    \tilde{r} = \max_{j \in [k]} \tilde{r}^{(j)} \qquad \text{and} \qquad 
    \tilde{r}^{(j)} = \max_{E \in \cP_{N^{(j)}}} E_{\mathrm{eff}}^{(j)}(E)  \ .
\end{align}
For example, we have~$\tilde{r}=1$ if all target states~$\Phi^{(j)}$ are Bell states. Improved noise strength of~$E_{\mathrm{eff}}$ is obtained by replacing~$r$ with $\tilde{r}$ in Eqs.~\eqref{eq:strengthofEffgeneral} and~\eqref{eq:strengthofEffspecific}.

\subsection{Generating long-range entangled Bell states using a parallel quantum bus\label{sec:busmultiple}}

Combining Corollary~\ref{cor:quantumbus} with Theorem~\ref{thm:parallelrepetition}, we have the following immediate consequence on the fault-tolerance of a parallel quantum bus:
\begin{theorem}[Quantum bus architecture (parallel)]\label{thm:parallelbuscomp}
There is a constant threshold~$p_0>0$ on error strength and constants $C>0$, $c\in (0,1)$ such that the following holds.
Let $k$ be arbitrary, and let $R_1,\ldots,R_k$ 
and $\Delta_1,\ldots,\Delta_k$ be such that 
\begin{align}
    \Delta_j&\geq 8 \log R_j\qquad\textrm{ for all }\qquad j\in [k]\ .
\end{align}
Consider the parallel bus~$\pi(\Delta_1,R_1)\times\cdots\times\pi(\Delta_k,R_k)$. Then a noisy
implementation of this parallel bus with noise strength~$p\leq p_0$ produces $k$~Bell states corrupted by local stochastic noise: The output state on qubits~$\{Q_{1,j},Q_{2,j}
\}_{j=1}^k$ is
\begin{align}
F\left(\bigotimes_{j=1}^k \ket{\Phi}_{Q_{1,j}Q_{2,j}}\right)
\end{align}
where $F$ is local stochastic noise on these $2k$~qubits with~$F\sim\cN(Cp^c)$.
\end{theorem}

\section{Fault-tolerantly localized quantum circuits} \label{sec:faulttolerantlylocalizing}
We saw in Section~\ref{sec:localization} how a general adaptive quantum circuit can be rendered geometrically local. The new circuit includes auxiliary qubits such that the circuit on the entire system is local in a 2D or 3D grid structure. In the new circuit, all ``long-range'' operations on two qubits with long distance are emulated by consuming short-range Bell pairs aligned on the routing paths.  Since the length of these paths linearly grows in the linear size of the grid, so does the number of short-range Bell pairs for each path. In the case where operations are noisy, this implies that the success probability of this long-range operation decays exponentially with the system size.

This problem can be resolved by augmenting the grid structure with quantum buses. We will use the result from Section~\ref{sec:busmultiple} that multiple pairs of noisy Bell states affected by constant-strength noise can be prepared by running multiple noisy quantum buses in parallel. In particular, such buses can create long-range entanglement, which 
can be used to (approximately) realize long-range operations. 

In Section~\ref{sec:ftpairwiseentanglementgeneration}, we describe the grid graphs we use and explain how to generate pairwise entanglement (respectively pair qubits) fault-tolerantly. 
 In Section~\ref{sec:thresholdtheorementirecircuit}, we state our main threshold theorem for fault-tolerantly localizing quantum circuits. As an application, we show in Section~\ref{sec:ftquantumcomputation} that any quantum computation using polynomial number of qubits and having polynomial depth can be fault-tolerantly localized with polynomial qubit overhead and quasi-polylogarithmic time overhead.

\subsection{Fault-tolerant pairwise entanglement generation in grid graphs}\label{sec:ftpairwiseentanglementgeneration}
In the following, it will be convenient to use the set $\mathbb{Z}_L=\{0,\ldots,L-1\}$ instead of~$[L]=\{1,\ldots,L\}$ to label the vertices of~$P_L$. 
The grid graph~$P_L\times P_L\times P_{4L}$ then has vertex set~$\mathbb{Z}_L\times\mathbb{Z}_L\times\mathbb{Z}_{4L}$, and we
consider the canonical embedding of this graph into~$\mathbb{R}^3$.  Recall that for $L$ even, the subset $\mathbb{Z}_L^2\times \{0\}$
can be parallel-routed with paths of length upper bounded by~$10L$, see Theorem~\ref{thm:main3Drouting}.
 One key property of these routing schemes we use below is the following: Each path~$\pi$ can decomposed into a collection~$(\pi^{\textrm{up}},{\pi_1^{\textrm{mid}}},{\pi_2^{\textrm{mid}}},\pi^{\textrm{down}})$
of four straight line segments that are parallel to different coordinate axes,
and have the property that adjacent line segments are orthogonal. 
This follows by construction, see Eq.~\eqref{eq:updownpathdefm}
for the definition of~$\pi^{\textrm{up}}$ and $\pi^{\textrm{down}}$,
and~\eqref{eq:pathinthesamefloor}
for the definition of the path~$\pi^{\textrm{mid}}$ which is the composed of segments~${\pi_1^{\textrm{mid}}}$ and ${\pi_2^{\textrm{mid}}}$.

We will show how to fault-tolerantly route $L^2$~qubits in a 3D architecture with additional qubits. These additional qubits are placed in such a way that a fault-tolerant quantum bus connecting the endpoints of each path segment can be realized, for any chosen path~$\pi$ occurring in the routing scheme (i.e., for some pairing of the vertices~$\mathbb{Z}_L^2\times \{0\}$).  Mirroring the requirement of edge-disjointness of the paths for a given pairing in a routing scheme, buses associated with any edge-disjoint collection of paths (respectively corresponding path segments) do not  share any qubits, thus enabling parallel operation. Since each path consists of only four segments, the associated four buses can be ``connected'' using entanglement swapping, thus producing entanglement at the endpoints of each path.  

Consider the set
\begin{align}
V(m,R)&:=\left\{\left(\frac{x}{m},\frac{y}{m},\frac{z}{m}\right)\ |\ x,y\in \{0,\ldots,Lm-1\},z\in \{0,\ldots,4Lm-1\right\}\ .
\end{align}
These $4(Lm)^3$ points correspond to the vertices of the grid graph~$P_{Lm }\times P_{Lm}\times P_{4Lm}$; they are the result of using the standard embedding of this graph in~$\mathbb{R}^3$ and subsequently rescaling by a factor~$\frac{1}{m}$. With a slight abuse of notation, we will refer to this embedded graph as~$\frac{1}{m}\left(P_{Lm}\times P_{Lm}\times P_{4m}\right)$.

We note that the vertices of $P_L\times P_L\times P_{4L}$ in the standard embedding are exactly the elements of~$V(m,R)$ with integer coordinates. Our construction uses three qubits $R_\ell^{\textrm{red}}$, $R_\ell^{\textrm{green}}$ and $R_\ell^{\textrm{blue}}$  at each location $\ell\in V(m,R)$. We use colors  to illustrate these. The main result of this section is the following:
\begin{theorem}[Fault-tolerant pairwise entanglement generation in a 3D grid graph]\label{thm:3DFTrouting}
There is a constant threshold error strength~$p_0>0$ and constants~$C>0$, $c\in (0,1]$ such that the following holds.
Let $L\geq 2$ be  an integer, and set $m:=82\lceil \log L\rceil$. 
Consider the embedded graph $\frac{1}{m}(P_{Lm}\times P_{Lm}\times P_{4Lm})$ with three qubits associated with each vertex.
Let  us write $\mathbb{Z}^2_L\times \{0\}=\{v_1,\ldots,v_{L^2}\}$ and let  $\{(v_{i_r},v_{j_r})\}_{r=1}^{L^2/2}$ be an arbitrary pairing
of these sites. 
Then there is an adaptive circuit~$\pi$ with the following properties:
\begin{enumerate}[(i)]
\item
The circuit~$\pi$ is constant-depth with local gates on $\frac{1}{m}(P_{Lm}\times P_{Lm}\times P_{4Lm})$.
\item
 The output of a noisy implementation of~$\pi$ with noise-strength~$p\leq p_0$ produces~$L^2/2$ Bell states corrupted by local stochastic noise: The output state on 
 the qubits
 \begin{align}
     \{R^{(\mathrm{blue})}_{v_{i_r}},R^{(\mathrm{blue})}_{v_{j_r}}\}_{r=1}^{L^2/2}
 \end{align}
 is 
 \begin{align}
 F\left(\bigotimes_{r=1}^{L^2/2} \ket{\Phi}_{R^{(\mathrm{blue})}_{v_{i_r}}R^{(\mathrm{blue})}_{v_{j_r}}}\right)\ 
 \end{align}
 where $F$ is local stochastic noise on these~$L^2$ qubits with $F\sim \cN(Cp^c)$. 
\end{enumerate}
\end{theorem}

\begin{figure}
    \centering
    \includegraphics[width=0.8\textwidth]{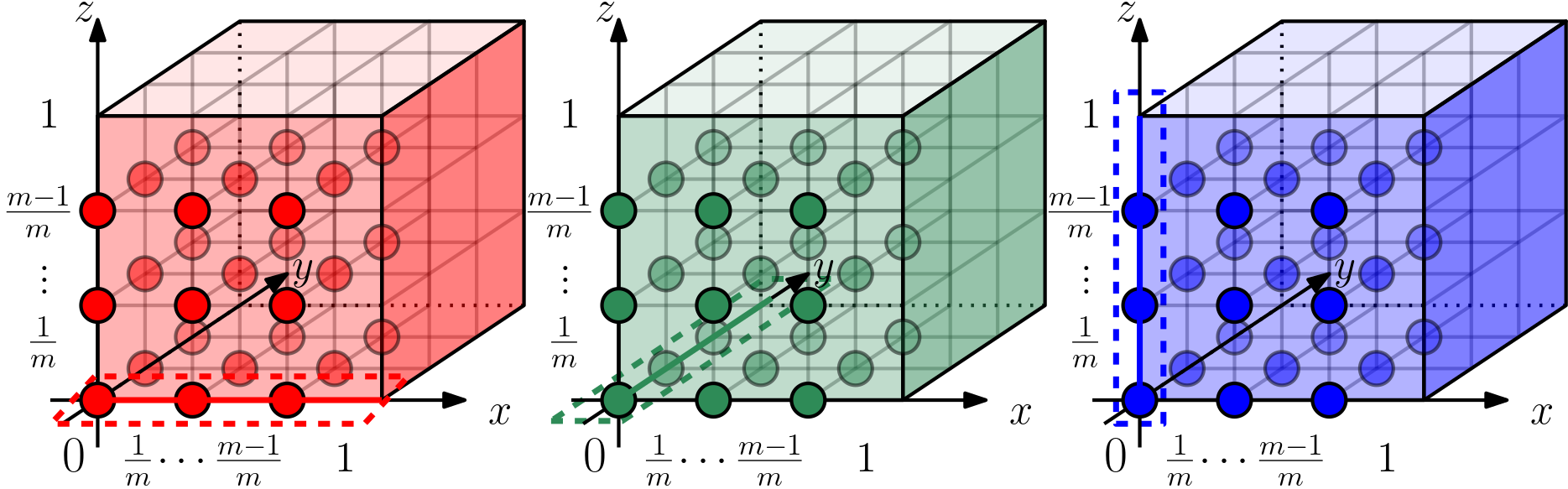}
\caption{Cubes~$\cQ(e)$ associated with edges~$e$. Each cube has $m$~qubits on each side, with spacing~$\frac{1}{m}$. Edges~$e$ are highlighted with dashed boxes. \label{fig:cubedefinition}}
\end{figure}

\begin{figure}
     \centering
     \begin{subfigure}[b]{0.45\textwidth}
         \centering
         \includegraphics[width=\textwidth]{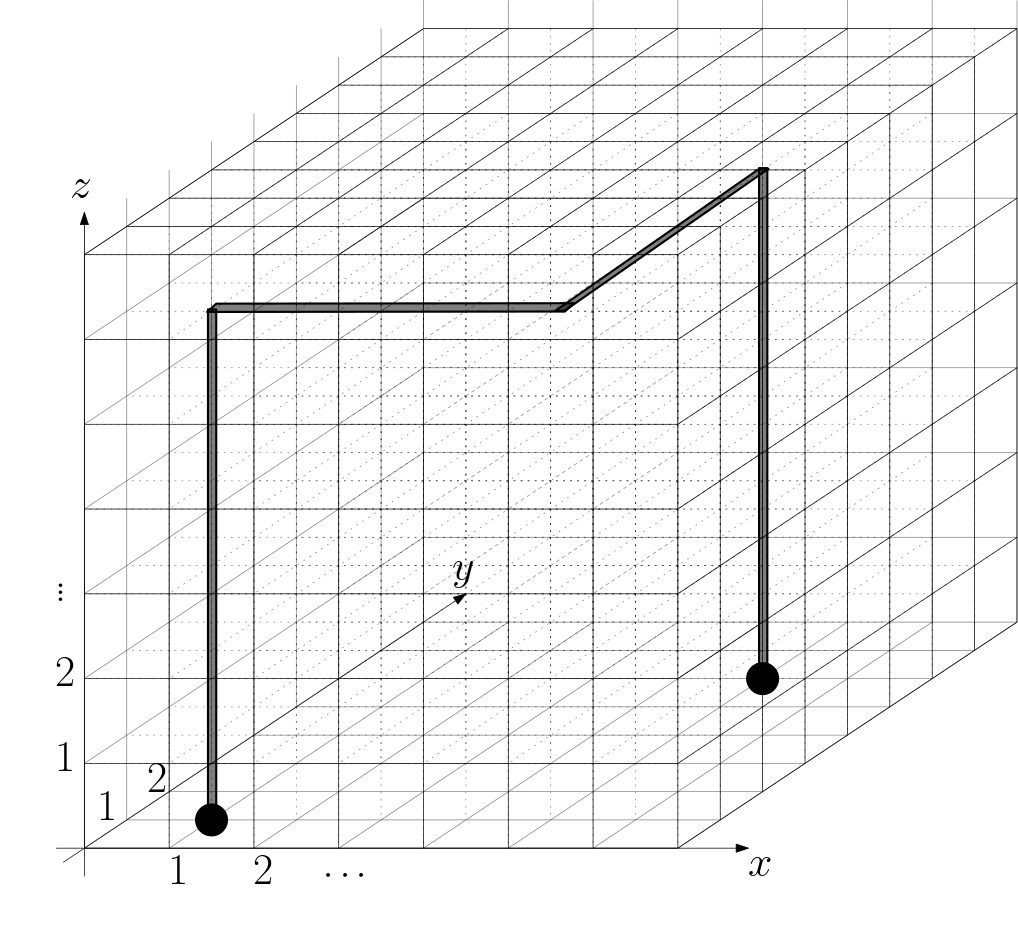}
         \caption{A path~$\pi_r$ connecting two vertices~$v_{i_r}, v_{j_r}$. The path consists of four line segments~$\{\pi_r^{(\alpha)}\}_{\alpha=0}^3$.}
         \label{fig:singlepath}
     \end{subfigure}
     \hfill
     \begin{subfigure}[b]{0.45\textwidth}
         \centering
         \includegraphics[width=\textwidth]{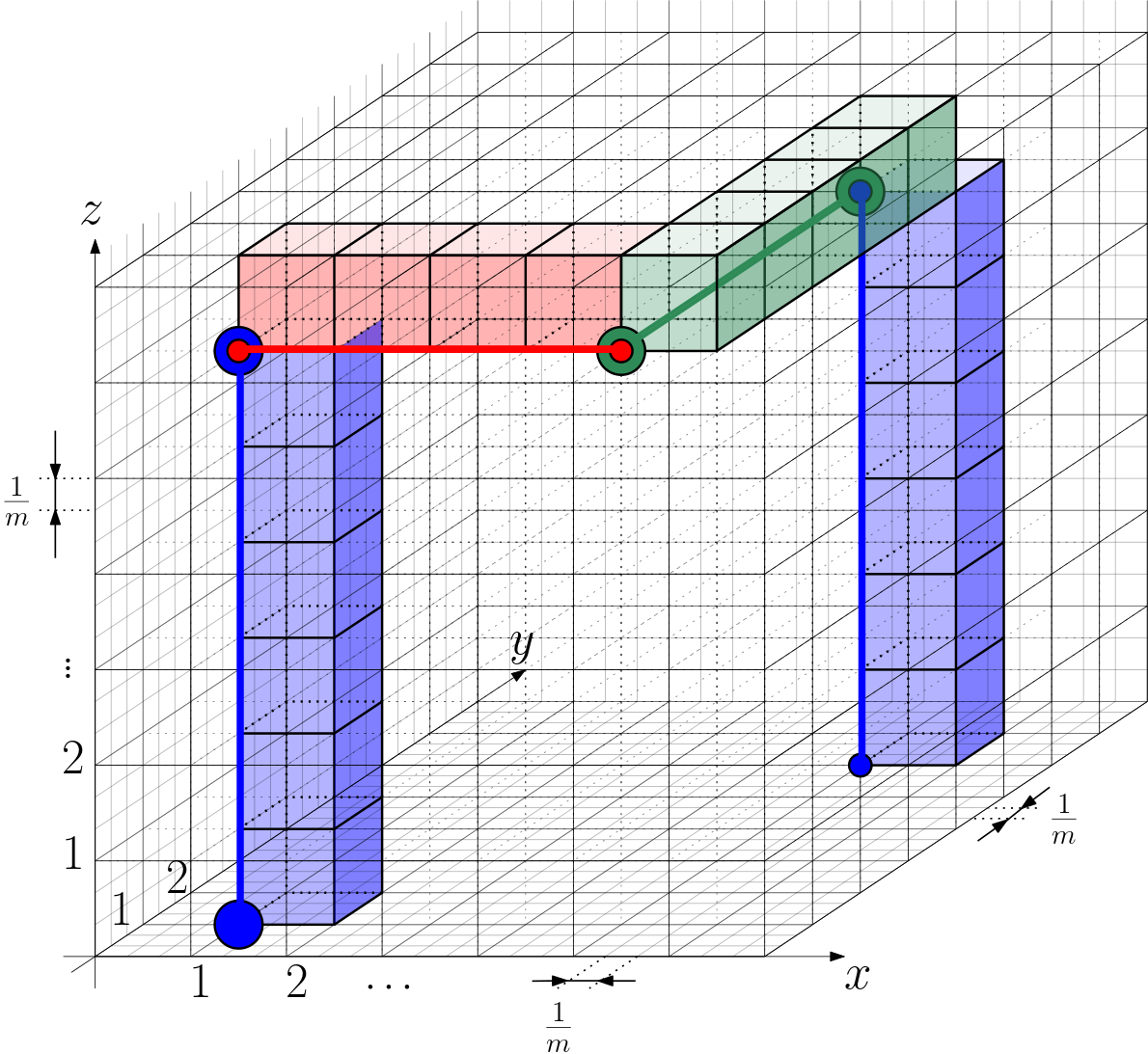}
         \caption{Four buses~$\{\cB_r^{(\alpha)}\}$ generating four Bell pairs on two qubits located at the endpoints of $\{\pi_r^{(\alpha)}\}_{\alpha=1}^3$.}
         \label{fig:quantumbuspath}
     \end{subfigure}
        \caption{A path~$\pi_r = \{\pi_r^{(\alpha)}\}_{\alpha=0}^3$ and four buses $\{\cB_r^{(\alpha)}\}_{\alpha=0}^3$ for fault-tolerant entanglement generation.}
        \label{fig:pathandbuspath}
\end{figure}

\begin{proof}
Let us associate a cube $\cQ(e)$ of qubits to each edge~$e$ of the graph~$P_L\times P_L\times P_{4L}$, see Fig.~\ref{fig:cubedefinition}. We distinguish between three different cases, depending on the orientation of the edge, i.e., whether it is parallel to the $X$-, the $Y$- or the $Z$-axis. That is, denoting by $\cQ_{(x,y,z)}:=[x,x+1)\times [y,y+1)\times [z,z+1)$ the ``half-open'' cube located at~$(x,y,z)$, we set
\begin{align}
\cQ(e)&:=
\begin{cases}
\left\{R^{\textrm{red}}_{\ell}\ |\ 
\ell\in \cQ_{(x,y,z)}\cap 
V(m,R)\right\} & \textrm{ if }e=\{(x,y,z),(x+1,y,z)\}\\
\\
\left\{R^{\textrm{green}}_{\ell}\ |\ 
\ell\in \cQ_{(x,y,z)}\cap 
V(m,R)\right\} & \textrm{ if }e=\{(x,y,z),(x,y+1,z)\}\\
\\
\left\{R^{\textrm{blue}}_{\ell}\ |\ 
\ell\in \cQ_{(x,y,z)}\cap 
V(m,R)\right\} & \textrm{ if }e=\{(x,y,z),(x,y,z+1)\} \ .
\end{cases}
\end{align}
Observe that  for two orthogonal edges~$e,e'$, the two cubes~$\cQ(e)$ and $\cQ(e')$
have no intersection since the associated qubits have different colors.  For colinear adjacent edges~$e$, $e'$, the intersection~$\cQ(e)\cap \cQ(e')$ is empty due to half-openness. Therefore, the  cubes~$\{Q(e)\}_{e}$ are pairwise disjoint.

Let now $\{\pi_r\}_{r=1}^{L^2/2}$ be a parallel routing scheme 
associated with a pairing~$\{(v_{i_r},v_{j_r})\}_{r=1}^{L^2/2}$  in the graph $P_L\times P_L\times P_{4L}$ as 
 provided by Theorem~\ref{thm:main3Drouting}; see Fig.~\ref{fig:singlepath} for  an illustration. 
 By  the disjointness of the cubes (and the edge-disjointness of the collection~$\{\pi_r\}_{r=1}^{L^2/2}$ of paths),  any two distinct paths~$\pi_r\neq \pi_{r'}$ in this routing scheme  have the property that  the collection of qubits in cubes associated with their edges are disjoint, i.e., 
\begin{align}
\left(\bigcup_{e\in \pi_r}\cQ(e)\right)\cap \left(\bigcup_{e\in \pi_{r'}}\cQ(e)\right)=\emptyset\ .
\end{align}
This implies that we can consider each path~$\pi_r$, $r\in [L^2/2]$ separately, and seek to generate a Bell pair on the qubits~$Q^{\textrm{blue}}_{v_{i_r}},Q^{\textrm{blue}}_{v_{j_r}}$ located at the endpoints using the qubits belonging to~$\bigcup_{e\in \pi_r}\cQ(e)$. 

In more detail, we will proceed as follows:
For each $r\in [L^2/2]$, we decompose the path~$\pi_r$ into
 four straight line segments~$(\pi_r^{(0)},\pi_r^{(1)},\pi_r^{(2)},\pi_r^{(3)})$ (as discussed above), and use
 a linear quantum bus~$\cB_r^{(\alpha)}$ for each $\alpha\in \{0,1,2,3\}$; see Fig.~\ref{fig:quantumbuspath}.
 The bus~$\cB_r^{(\alpha)}$ creates a Bell pair on two qubits located at the endpoints~$\partial \pi_r^{(\alpha)}=:\{u_{r}^{(\alpha)},w_{r}^{(\alpha)}\}$.  
 Here we denote by $u_r^{(\alpha)}$ and $w_{r}^{(\alpha)}$ the first and last vertices of~$\pi_r^{(\alpha)}$, where the orientation of the path segment~$\pi_r^{(\alpha)}$ is defined by a traversal of the path~$\pi_r$ from~$v_{i_r}$ to $v_{j_r}$.
The bus~$\cB_r^{(\alpha)}$ uses the collection
\begin{align}
\cQ_r^{(\alpha)}=\bigcup_{e\in \pi_r^{(\alpha)}} \cQ(e)
\end{align}
of qubits associated with edges belonging to the line segment~$\pi_r^{(\alpha)}$. Recall that all these qubits have the same color: It is the color determined by the coordinate axis the line segment~$\pi^{(\alpha)}_r$ is parallel to.  We will denote by $S_r^{(\alpha)}$ the qubit of this color at the beginning vertex~$u_r^{(\alpha)}$ of the line segment, and by~$T_r^{(\alpha)}$
the qubit  of the same color at the ending vertex~$w_r^{(\alpha)}$.

Note that the qubits~$\cQ_r^{(\alpha)}$  are located on the vertices of an embedded  version of the grid graph~$P_{m}\times P_m\times P_{R^{(\alpha)}_r}$, where 
the length~$R^{(\alpha)}_r=m\cdot |\pi_r^{(\alpha)}|$  (i.e., number of qubits) is proportional to the length~$|\pi_r^{(\alpha)}|$ of the path segment. (Compared to the standard embedding, this involves scale factor~$\frac{1}{m}$ and a rotation such that the graph is parallel to~$\pi_r^{(\alpha)}$.)  Since this has the form of a ``standard'' (linear) quantum bus on the grid graph~$P_{\Delta}\times P_{\Delta}\times P_R$, we can  take the protocol~$\cB_r^{(\alpha)}$ to be the associated scheme~$\pi(\Delta,R)$, that is, the protocol discussed in Corollary~\ref{cor:quantumbus}. By Corollary~\ref{cor:quantumbus}, this protocol generates entanglement at the endpoints of~$\pi_r^{(\alpha)}$ if
\begin{align}
m \geq 8 \log (m \cdot |\pi_r^{(\alpha)}|)\ .\label{eq:delta8deltapr}
\end{align}
Since $|\pi_r^{(\alpha)}|\leq |\pi_r|\leq 10L$ for each path segment (see Theorem~\ref{thm:main3Drouting}), Eq.~\eqref{eq:delta8deltapr}
is satisfied for every path segment if
\begin{align}
m \geq 8 \log(10 m L)\ .\label{eq:eightmlzen}
\end{align}
Eq.~\eqref{eq:eightmlzen} is satisfied because of our definition $m:=82\lceil \log L\rceil$ and $L\geq 2$. Thus each bus~$\cB_r^{(\alpha)}$ is $(p_0,f(p)=p/p_0)$-robust with $p_0=1/5004$, and generates a Bell pair at the endpoints of~$\pi_r^{(\alpha)}$.

In summary, we have a total of~$4\cdot \frac{L^2}{2}$ buses~$\{\cB_r^{(\alpha)}\}_{j\in \{0,1,2,3\},r\in [L^2/2]}$ that are each~$(p_0,f(p)=p/p_0)$-robust. 
The parallel repetition Theorem~\ref{thm:parallelbuscomp} thus guarantees that if these protocols are run in parallel the result is a noisy version 
\begin{align}
 F'\left(\bigotimes_{r=1}^{L^2/2} 
\bigotimes_{\alpha=0}^3\ket{\Phi}_{S^{(\alpha)}_rT^{(\alpha)}_r}\right)\ 
 \end{align}
 of $4(L^2/2)$-Bell pairs (each spanning a line segment), where $F'\sim\cN(C'p^{c'})$ is local stochastic for some constants~$C',c'$. 
 
Finally, we apply entanglement-swapping  in parallel to each path~$\pi_r$, $r\in [L^2/2]$.  That is, for each $r\in [L^2/2]$, we perform three Bell measurements
on the three qubits pairs 
$(T_{r}^{(0)},S_{r}^{(1)})$, $(T_{r}^{(1)},S_{r}^{(2)})$ and $(T_{r}^{(2)},S_{r}^{(3)})$. This 
converts the state~$\bigotimes_{\alpha=0}^3\ket{\Phi}_{S_r^{(\alpha)}T_r^{(\alpha)}}$
to a Bell state~$\ket{\Phi}_{S_r^{(0)}T_r^{(3)}}$ (after applying a Pauli correction depending on the measurement result).

By definition, the path segments~$\pi_r^{(0)}$, $\pi_r^{(3)}$ are  aligned with the $Z$-axis, i.e., correspond to blue qubits. In particular, $(S_r^{(0)}, T_r^{(3)})=(R^{\textrm{blue}}_{v_{i_r}},R^{\textrm{blue}}_{v_{j_r}})$ are the two blue qubits at the endpoints of~$\pi_r$, for each $r\in [L^2/2]$.
Together with  Lemma~\ref{lem:noisyentanglementswapping} (which demonstrates that entanglement swapping in parallel is compatible with local stochastic noise), this shows that the final state on these qubits after the entanglement-swapping step is the desired target state
\begin{align}
 F\left(\bigotimes_{r=1}^{L^2/2} \ket{\Phi}_{R^{\textrm{blue}}_{v_{i_r}}R^{\textrm{blue}}_{v_{j_r}}}\right)\ 
 \end{align}
up to a local stochastic error~$F\sim\cN(Cp^c)$
 for some constants~$C,c>0$. This concludes the proof of 
Theorem~\ref{thm:3DFTrouting}. 
\end{proof}

By analogous reasoning, we can give a protocol for fault-tolerant pairwise entanglement generation in a quasi-2D architecture. Here 
quasi-2D refers to the fact that the linear extent along one of the coordinate axes is only logarithmic in the linear extent along the other two axes.
\begin{theorem}[Quasi-2D-local fault-tolerant pairwise entanglement generation]\label{thm:FTroutingquasi2D}
There is a constant threshold error strength~$p_0>0$ and constants~$C>0$, $c\in (0,1]$ such that the following holds.
Let $L\geq 2$ be  an integer. 
Set $m:=82\lceil \log L\rceil$. 
Consider the embedded graph $\frac{1}{m}(P_{Lm}\times P_{Lm}\times P_{m})$ with two qubits associated with each vertex.
Let us write $\{(i,i,0) \mid i \in \mathbb{Z}_L\} = \{v_1,\ldots,v_{L}\}$ and
let $\{(v_{i_r},v_{j_r})\}_{r=1}^{L/2}$ be an arbitrary pairing of these sites. 
Then there is an adaptive circuit~$\pi$ with the following properties:
\begin{enumerate}[(i)]
\item
The circuit~$\pi$ is constant-depth with local gates on $\frac{1}{m}(P_{Lm}\times P_{Lm}\times P_{m})$.
\item
 The output of a noisy implementation of~$\pi$ with noise-strength~$p\leq p_0$ produces $L/2$~Bell states corrupted by local stochastic noise: The output state on 
 the qubits 
 \begin{align}
  \{R^{(\mathrm{red})}_{v_{i_r}},R^{(\mathrm{green})}_{v_{j_r}}\}_{r=1}^{L/2}   
 \end{align}
  is 
 \begin{align}
 F\left(\bigotimes_{r=1}^{L/2} \ket{\Phi}_{R^{(\mathrm{red})}_{v_{i_r}}R^{(\mathrm{green})}_{v_{j_r}}}\right)\ 
 \end{align}
 where $F$ is local stochastic noise on these~$L$ qubits with $F\sim \cN(Cp^c)$. 
\end{enumerate}
\end{theorem}
\begin{proof}
    The proof is similar to that of Theorem~\ref{thm:3DFTrouting}, so we only provide a sketch.
    We use the parallel routing~$\{\pi_r\}_{r=1}^{L/2}$ of the pairing~$\{(v_{i_r}, v_{j_r})\}_{r=1}^{L/2}$ obtained by Lemma~\ref{lem:parallelrouting2D}, instead of that from Theorem~\ref{thm:main3Drouting}. Note that each path~$\pi_r$ is decomposed into two  orthogonal line segments~$\pi_r^{(0)}, \pi_r^{(1)}$.  For fault-tolerance, it suffices to show that the inequality~\eqref{eq:delta8deltapr} holds. This is the case since $\abs{\pi_r^{(\alpha)}} \leq 2L$ for $\alpha \in \{0,1\}$ (see Lemma~\ref{lem:parallelrouting2D}).
\end{proof}

As a consequence of Theorem~\ref{thm:3DFTrouting}, the 
fault-tolerance of parallel quantum teleportation (see Lemma~\ref{lem:noisyquantumteleportation}), and the fact that local stochastic noise remains local stochastic under a depth-$1$ circuit, we obtain a fault-tolerant qubit pairing circuit with local operations in 3D. For this, we add two qubits $\{Q_j, P_j\}$ at each site~$v_j \in \mathbb{Z}_L^2 \times \{0\}$ in the grid graph~$\frac{1}{m} P_{Lm} \times P_{Lm} \times P_{4Lm}$.

\begin{corollary}[3D-local fault-tolerant qubit pairing]\label{cor:FTqubitpair3D}
There is a constant threshold error strength~$p_0>0$ and constants~$C>0$, $c\in (0,1]$ such that the following holds.
Let $L\geq 2$ be  an integer and set $m:=82\lceil \log L\rceil$. 
Consider the embedded graph $\frac{1}{m}(P_{Lm}\times P_{Lm}\times P_{4Lm})$ with three qubits associated with each vertex and additional qubits~$\{Q_j, P_j\}$ at each site~$v_j$.
Let us write $\mathbb{Z}^2_L\times \{0\}=\{v_1,\ldots,v_{L^2}\}$, and let $\{(v_{i_r},v_{j_r})\}_{r=1}^{L^2/2}$ be an arbitrary pairing
of these sites. 
Then there is a constant-depth adaptive circuit~$\cQ_{\mathrm{FTpair,3D}}$ with the following properties: the circuit~$\cQ_{\mathrm{FTpair,3D}}$ implements a transfer of subsystems that maps each subsystem~$Q_{j_r}$ to the subsystem~$P_{i_r}$ for $r \in [L^2/2]$, i.e., an $L^2$-qubit state~$\Psi$ on the registers~$Q_1, \dots, Q_{L^2}$ is mapped according to
\begin{align}
\Psi_{Q_{i_1}Q_{j_1}\cdots Q_{i_{L^2/2}}Q_{j_{L^2/2}}} & \mapsto \Psi_{Q_{i_1}P_{i_1}\cdots Q_{i_{L^2/2}}P_{i_{L^2/2}}}\ .\label{eq:actionofcircuitqftpair3d}
\end{align}
Moreover, for any noisy implementation of the circuit~$\cQ_{\mathrm{FTpair}}$ with local stochastic noise of strength~$p \leq p_0$, the resulting state on the registers~$Q_{i_1}P_{i_1} \dots Q_{i_{L^2/2}} P_{i_{L^2/2}}$ is
\begin{align}
    F \Psi_{Q_{i_1}P_{i_1}\cdots Q_{i_{L^2/2}}P_{i_{L^2/2}}}
\end{align}
for local stochastic noise~$F \sim \cN(Cp^c)$.
\end{corollary}

Following the same reasoning, we obtain a fault-tolerant qubit pairing circuit with local operations in a quasi-2D architecture as a consequence of Theorem~\ref{thm:FTroutingquasi2D}. Here we add two qubits~$\{Q_j, P_j\}$ at each site~$v_j \in \{(i,i,0) \mid i \in \mathbb{Z}_L\}$ in the grid graph~$\frac{1}{m}P_{Lm} \times P_{Lm} \times P_{m}$.
\begin{corollary}[Quasi-2D-local fault-tolerant qubit pairing]\label{cor:FTqubitpairquasi2D}
There is a constant threshold error strength $p_0>0$ and constants~$C>0$, $c\in (0,1]$ such that the following holds.
Let $L\geq 2$ be  an integer, and set $m:=82\lceil \log L\rceil$. 
Consider the embedded graph $\frac{1}{m}(P_{Lm}\times P_{Lm}\times P_{4Lm})$ with two qubits associated with each vertex and additional qubits~$\{Q_j, P_j\}$ at each site~$v_j$. 
Let us write $\{(i,i,0) \mid i \in \mathbb{Z}_L\}=\{v_1,\ldots,v_{L}\}$, and let $\{(v_{i_r},v_{j_r})\}_{r=1}^{L/2}$ be an arbitrary pairing of these sites. 
Then there is a constant-depth adaptive circuit~$\cQ_{\mathrm{FTpair,quasi2D}}$ with the following properties: the circuit~$\cQ_{\mathrm{FTpair,quasi2D}}$ implements a transfer of subsystems that maps each subsystem~$Q_{j_r}$ to the subsystem~$P_{i_r}$ for $r \in [L/2]$, i.e., an $L$-qubit state~$\Psi$ on the registers~$Q_1, \dots, Q_{L}$ is mapped according to
\begin{align}
\Psi_{Q_{i_1}Q_{j_1}\cdots Q_{i_{L/2}}Q_{j_{L/2}}} & \mapsto \Psi_{Q_{i_1}P_{i_1}\cdots Q_{i_{L/2}}P_{i_{L/2}}}\ .\label{eq:actionofcircuitqftpairquasi2d}
\end{align}
Moreover, for any noisy implementation of the circuit~$\cQ_{\mathrm{FTpair,quasi2D}}$ with  local stochastic noise of strength~$p \leq p_0$, the resulting state on the registers~$Q_{i_1}P_{i_1} \dots Q_{i_{L/2}} P_{i_{L/2}}$ is
\begin{align}
    F \Psi_{Q_{i_1}P_{i_1}\cdots Q_{i_{L/2}}P_{i_{L/2}}}
\end{align}
for local stochastic noise~$F \sim \cN(Cp^c)$.
\end{corollary}

\subsection{Threshold theorem of quantum circuits localized by a parallel bus}\label{sec:thresholdtheorementirecircuit}
Having constructed fault-tolerant qubit pairing protocols, we are now in position to state and prove our main result. It states the functionality of a fault-tolerant circuit~$\cQ$ (which uses non-local gates) can be  realized by local circuits.
\begin{theorem}[Fault-tolerantly localizing a general adaptive circuit]\label{thm:localizingcircuit}
Let $\cQ$ be a general adaptive quantum circuit on $n=2k$~qubits. Let $T$ be the quantum depth of~$\cQ$.  Then there are adaptive quantum circuits~$\cQ'_{\mathrm{quasi2D}}$ and $\cQ'_{\mathrm{3D}}$ with the following properties:
\begin{enumerate}[(i)]
\item
By taking certain marginals (i.e., tracing out qubits and/or ignoring measurement results), the two circuits exactly simulate~$\cQ$.
\item Both circuits have quantum depth of order~$O(T)$.
\item
The circuits $\cQ'_{\mathrm{quasi2D}}$ and $\cQ'_{\mathrm{3D}}$ are geometrically local in quasi-2D and 3D-architectures, respectively (i.e., only involve local or nearest-neighbor operations on a corresponding grid graph). 
\item
The circuits~$\cQ'_{\mathrm{quasi2D}}$ and~$\cQ'_{\mathrm{3D}}$ use a total number of
\begin{align}
n^{\mathrm{tot}}_{\mathrm{quasi2D}} &= O(n^2\log^3 n)\qquad\textrm{ and }\qquad n^{\mathrm{tot}}_{\mathrm{3D}}=O(n^{3/2}\log^3 n)\label{eq:numberofqubitslocalizedcircuitsft}
\end{align}
qubits, respectively.
 \item 
    There exist a threshold~$p_0 > 0$ and some constants~$C > 0,  c > 0$ such that the following holds: Any noisy implementation of the circuits~$\cQ'_{\mathrm{quasi\text{-}2D}}$ and~$\cQ'_{\mathrm{3D}}$ under  arbitrary local stochastic noise of strength~$p \leq p_0$ are equivalent to a noisy implementation of the circuit~$\cQ$ with local stochastic noise of strength $C\cdot p^c$.
    \end{enumerate}
\end{theorem}
\begin{proof} 
We use the graphs
\begin{align}
G_{\mathrm{quasi2D}}&=\frac{1}{m} P_{Lm} \times P_{Lm} \times P_{m} \qquad\textrm{ and } \qquad G_{\mathrm{3D}}=\frac{1}{m}P_{Lm}\times P_{Lm}\times P_{4Lm}
\end{align}
with
\begin{align}
    L := \begin{cases}
        \lceil n\rceil & \text{in $G_{\mathrm{quasi2D}}$} \\
        \ceil{\sqrt{n}} & \text{in $G_{\mathrm{3D}}$}  
    \end{cases} \ , \qquad \text{and} \qquad
    m:=82\lceil\log L\rceil \ ,
\end{align}
see Theorem~\ref{thm:3DFTrouting} respectively Theorem~\ref{thm:FTroutingquasi2D}. For $G_{\mathrm{quasi2D}}$ we
will use the vertices
\begin{align}
S_{\mathrm{quasi2D}}=\{v_i:=(i-1,i-1,0)\}_{i\in [n]}\subset \mathbb{Z}_L^2\times \{0\}\ ,
\end{align}
whereas for the 3D~grid graph~$G_{\mathrm{3D}}$, we pick an  arbitrary subset 
\begin{align}
S_{\mathrm{3D}}&=\{v_i\}_{i=1}^n\subseteq\mathbb{Z}_L^2\times \{0\}
\end{align}
of~$\mathbb{Z}_L^2\times \{0\}$ of  size~$|S_{\mathrm{3D}}|=n\leq L^2$.  
(We will attach our ``computational'' qubits~$Q_1\cdots Q_n$ associated with the circuit~$\cQ$ to these vertices.)

Associating two and three auxiliary qubits respectively to each vertex of the graphs~$G_{\mathrm{quasi2D}}$ and~$G_{\mathrm{3D}}$, respectively, 
and attaching a ``computational'' qubit~$Q_j$ and a ``register'' qubit~$P_j$ to each vertex~$v_j$ for $j\in [n]$, we are using a total number of
\begin{align}
n^{\mathrm{tot}}_{\mathrm{quasi2D}} = 2n + 2 \cdot L^2 m^3  =O(n^2\log^3 n)\qquad\textrm{ and }\qquad n^{\mathrm{tot}}_{\mathrm{3D}}=2n+3\cdot4(Lm)^3=O(n^{3/2}\log^3n)
\end{align}
qubits, respectively. This agrees with~Eq.~\eqref{eq:numberofqubitslocalizedcircuitsft}. 

We construct our circuits~$\cQ'_{\mathrm{quasi2D}}$ and $\cQ'_{\mathrm{3D}}$ by replacing -- in each layer~$t\in [T]$ of the original adaptive circuit~$\cQ$ -- the operation~$\cM^{(t)}$ (see Eq.~\eqref{eq:generallayerm}) by a geometrically local adaptive circuit~$\widehat{\cM}^{(t)}$. 
Let~$\{(i_r^{(t)},j_r^{(t)})\}_{r=1}^k$ be the pairing relevant at time step (layer)~$t\in [T]$, and let 
\begin{align} \label{eq:qpairtpairft} 
\cQ^{(t)}_{\textrm{FTpair}}=\cQ_{\textrm{FTpair}}\left(\{(i_r^{(t)},j_r^{(t)})\}_{r=1}^k\right), \qquad \text{and} \qquad (\cQ^{(t)}_{\textrm{FTpair}})^{-1}=\cQ^{-1}_{\textrm{FTpair}}\left(\{(i_r^{(t)},j_r^{(t)})\}_{r=1}^k\right) 
\end{align}
be the constant-depth adaptive circuits introduced in Section~\ref{sec:ftpairwiseentanglementgeneration}, i.e., the circuit $\cQ_{\mathrm{FTpair}}$ is either $\cQ_{\mathrm{FTpair,quasi2D}}$ or $\cQ_{\mathrm{FTpair,3D}}$.  Then the new circuits~$\cQ_{\mathrm{quasi2D}}'$ respectively~$\cQ_{\mathrm{3D}}'$ are obtained by replacing -- for each~$t\in [T]$ -- the operation
~$\cM^{(t)}$ in the circuit~$\cQ$ 
by the composition
\begin{align}\label{eq:qftpair}
\widehat{\cM}^{(t)}:=
(\cQ^{(t)}_{\textrm{FTpair}})^{-1}\circ
\left(\bigotimes_{r=1}^k \cM^{(t,r)}_{Q_{i^{(t)}_r}P_{i^{(t)}_r}}\right)
\circ \cQ^{(t)}_{\textrm{FTpair}}\ .
\end{align}
In other words, the qubits paired in layer~$t$ are placed next to each other by application of~$\cQ^{(t)}_{\textrm{FTpair}}$. Then  each (two-qubit) 
operation~$\cM^{(t,r)}$, $r\in [k]$ can be applied locally. Subsequently,  the qubits are 
moved back to their original positions by application of $(\cQ^{(t)}_{\textrm{FTpair}})^{-1}$.  The new circuits~$\cQ_{\textrm{quasi2D}}'$, $\cQ_{\textrm{3D}}'$ then are obtained as the composition
\begin{align}
\cQ'&:=\widehat{\cM}^{(T)}\circ \cdots \circ \widehat{\cM}^{(1)}\ .
\end{align} 
It is easy to check that this has all the claimed properties. In particular, for each $t\in [T]$, both circuits~\eqref{eq:qftpair} are adaptive constant-depth circuits, see Corollaries~\ref{cor:FTqubitpair3D} and~\ref{cor:FTqubitpairquasi2D}. 
This immediately implies that the quantum circuit depth of~$\cQ'$ is of order~$O(T)$. 
 \end{proof}

\subsection{Local quantum fault-tolerance in 3D\label{sec:ftquantumcomputation}}
Here we explain how Corollary~\ref{cor:faulttolerancelocal} follows from our construction when applied to the
fault-tolerance scheme of~\cite{yamasakiTimeEfficientConstantSpaceOverheadFaultTolerant2024}. 
The latter takes an ideal quantum circuit~$\cQ_{\mathrm{ideal}}$ and transforms it into a fault-tolerant (adaptive) circuit~$\cQ_{\mathrm{FT}}$ whose output distribution approximates that of the ideal circuit even when it is implemented imperfectly.
Paraphrased using our terminology, the main result of Yamasaki and Koashi gives the following.
\begin{theorem}[Fault-tolerant quantum computation with non-local gates~\cite{yamasakiTimeEfficientConstantSpaceOverheadFaultTolerant2024}]\label{thm:faulttolerancelocalresultyamasakikoashi}
There is a threshold error strength~$p_0>0$ such that the following holds for all sufficiently large~$n$ and an arbitrary constant~$\varepsilon \in (0,1)$. 
Let~$\cQ_{\mathrm{ideal}}$ be an adaptive quantum circuit using $n$~qubits and having quantum depth~$T(n)=O(\poly(n))$.
There is a circuit~$\cQ_{\mathrm{FT}}$ with the following properties:
\begin{enumerate}[(i)]
    \item 
    The circuit~$\cQ_{\mathrm{FT}}$ uses $O(n)$ qubits.
    \item The quantum depth of~$\cQ_{\mathrm{FT}}$ is of order $T(n) \cdot \exp(O(\log^2(\log(n/\varepsilon))))$.
    \item 
    Any noisy implementation of~$\cQ_{\mathrm{FT}}$ with local stochastic errors of strength~$p \leq p_0$ has an output distribution whose  total variation distance to the output distribution of~$\cQ_{\mathrm{ideal}}$  is upper bounded by~$\varepsilon$. 
\end{enumerate}
\end{theorem}
Applying our construction (that is, Theorem~\ref{thm:localizingcircuit}) to 
the circuit~$\cQ_{\mathrm{FT}}$  yields a 3D-local circuit~$\cQ'$ with properties as stated in Corollary~\ref{cor:faulttolerancelocal}, as desired.

\subsubsection*{Acknowledgements}
SC and RK gratefully acknowledge support by the European Research Council under grant agreement no.\ 101001976 (project EQUIPTNT).
SC thanks Sebastian Stengele for discussion on routing problems. 
\printbibliography
\newpage 
\appendix

\section{Local stochastic noise and simple  adaptive quantum circuits \label{sec:lemmasnoisycircuit}}
Local stochastic noise is
compatible with 
Clifford circuits, as expressed by the following Lemma, (cf. e.g.,~\cite[Lemma 11]{bravyiQuantumAdvantageNoisy2020}  and~\cite{fawziConstantOverheadQuantum2018}). 

\begin{lemma}[Properties of local stochastic noise] \label{lem:propertieslocalstochasticnoise}
Consider a system of~$n$ qubits. Then the following holds.
\begin{enumerate}[(i)]
    \item 
    Suppose $E \sim \cN(p)$ and $E'$ is a random Pauli such that $\supp(E') \subseteq \supp(E)$ with probability $1$. Then $E' \sim \cN(p)$. \label{it:subset}
    \item 
    Suppose $E \sim \cN(p)$ and $F \sim \cN(q)$ are random Paulis which may be dependent. Then $E \cdot F \sim \cN(q')$ where $q' = 2 \max \{\sqrt{p}, \sqrt{q}\}$. \label{it:productdependent}
    \item 
    Suppose $E \sim \cN(p)$ is a random Pauli and let $U$ be a tensor product of single- and two-qubit Clifford unitaries. Then $UEU^{\dagger} \sim\cN(q')$ where $q' = \sqrt{2p}$. 
    \label{it:clifford}
    \item 
    Let $A\subseteq [n]$ be a subset of qubits.
    Suppose $E \sim \cN(p)$, $F \sim \cN(q)$ and $\supp(E) \subseteq A$ and $\supp(F) \subseteq A^c=:[n]\backslash A$ with probability $1$. Then $E \cdot F \sim \cN(q')$ where $q' = \max \{\sqrt{p},\sqrt{q}\}$. \label{it:productdisjoint}
    \item 
    Let $\pi$ be a permutation on $[n]$, and let $\pi(E)$ be the $n$-qubit Pauli which acts on the $i$-th qubit as $\pi(i)$. Then $\pi(E) \sim \cN(p)$.
    \label{it:permutation}
\end{enumerate}
\end{lemma}
\begin{proof}
    See~\cite[Lemma 11]{bravyiQuantumAdvantageNoisy2020} for \eqref{it:subset}--\eqref{it:clifford}.
    For \eqref{it:productdisjoint}, observe that  for any subset $I \subseteq [n]$, we have
    \begin{align}
        \Pr[I \subseteq \supp(E \cdot F)]
        &= \Pr[I \cap A \subseteq \supp(E) \ \text{and} \ I \cap A^c \subseteq \supp(F)]\\
        &\leq \min\{p^{\abs{I \cap A}}, q^{\abs{I \cap A^c}}\} \\
        &\leq \min \{ \left(\max \{p,q\}\right)^{\abs{I \cap A}}, \left(\max \{p,q\}\right)^{\abs{I \cap A^c}} \} \\
        &\leq \left(\max\{p,q\}\right)^{\frac{\abs{I \cap A} + \abs{I \cap A^c}}{2}} \\
        &= \left(\max\{\sqrt{p},\sqrt{q}\}\right)^{\abs{I}} \, ,
    \end{align}
    where we used the definition of local stochastic noise and the fact that $\min \{a,b\} \leq \sqrt{ab}$ for any $a, b \geq 0$.
    For \eqref{it:permutation}, observe that for any subset $I = \{i_1, \dots, i_k\} \subseteq [n]$, we have
    $I \subseteq \supp(\pi(E))$ if and only if $\{\pi^{-1}(i_1), \dots, \pi^{-1}(i_k)\} =: \pi^{-1}(I) \subseteq \supp(E)$. Therefore,
    \begin{align}
        \Pr[I \subseteq \supp(\pi(E))] = \Pr[\pi^{-1}(I) \subseteq \supp(E)] \leq p^{\abs{\pi^{-1}(I)}} = p^{\abs{I}} \, , 
    \end{align}
    where the fact that $E$ is local stochastic noise is used. 
\end{proof}

Lemma~\ref{lem:propertieslocalstochasticnoise} allows one to commute errors through adaptive circuits explicitly  if the adaptivity is expressed by linear functions.  We formulate this as follows:
\begin{lemma}[Noisy adaptive quantum circuits] \label{lem:noisyadaptivecircuits}
    Let $n, n_1, n_2 \in \mathbb{N}$ with $n=n_1 + n_2$ and 
    let 
    \begin{align}
    A = (a_{ij}) \ , \  B = (b_{ij}) \in \mathbb{F}_2^{n_1 \times n_2} 
    \end{align}
    be matrices.
    Consider a set of qubits~$\cA = \cA_1 \cup \cA_2 = [n]$ where~$\cA_1 = [n_1]$ and~$\cA_2 = \{n_1+1, \dots, n\}$.
    For $x\in \mathbb{F}_2^n$
    let $X(x)=\prod_{j=1}^n X_j^{x_j}$
    and $Z(x)=\prod_{j=1}^n Z_j^{x_j}$
    be products of Pauli-$X$ and Pauli-$Z$ operators with support defined by~$x$. For each~$x_1 \in \mathbb{F}_2^{n_1}$, let us denote by~$X(x_1)$ and~$Z(x_1)$ Pauli operators on~$\cA_1$ defined analogously.
    
    Let $\cQ$ be an adaptive quantum circuit on the qubits~$\cA$ of the following:
    \begin{enumerate}
    \item 
The circuit~$\cQ$ measures all qubits belonging to~$\cA_2$ in the computational basis, obtaining a measurement outcome~$z\in \mathbb{F}_2^{n_2}$. 
\item Subsequently, an $n_1$-qubit Pauli gate~$C(z)\in \cP_{n_1}$
is applied to the qubits~$\cA_1$, where we assume that 
    \begin{align}
        C(z) = X(Az) Z(Bz) \qquad \text{for all measurement outcomes} \quad z \in \mathbb{F}_2^{n_2} \ .
    \end{align}
    \end{enumerate}
    For any Pauli operator~$E$ on~$\cA$, there exists a Pauli operator~$F=F(E)$ on~$\cA_1$ determined by~$E$ satisfying the following:
    \begin{enumerate}[(i)]
        \item \label{it:noisyEFequiv}
        A noisy implementation~$\cQ_1=\cQ\circ \cU_E$ of the circuit~$\cQ$ associated with the error~$E$ acting on~$\cA$ before the (ideal) circuit~$\cQ$ is applied is equivalent to a noisy implementation~$\cQ_2=\cU_F\circ\cQ$ of the circuit~$\cQ$ associated with the error~$F$ acting on~$\cA_2$ after the (ideal) circuit~$\cQ$ is applied (see Fig.~\ref{fig:noisyadaptivecircuitsequiv}.) 
        \item \label{it:Fstrength}
        Let $w \in \mathbb{N}$ with $n_1 \cdot w \leq n_2$, and
        suppose that all non-zero row vectors of~$A$ and~$B$ have weight $w$ and all column vectors of~$A$ and~$B$ have weight at most 1, i.e.,
        \begin{align}
        \begin{matrix}
            \sum_{j=1}^{n_2} a_{ij}&\in&\{0,w\} \ ,&\sum_{j=1}^{n_2} b_{ij}&\in&\{0,w\} &\qquad  \text{for all}& i \in[n_1] \\
            \sum_{i=1}^{n_1} a_{ij}&\leq& 1 \ , &\sum_{i=1}^{n_1} b_{ij}&\leq&1 &\qquad \text{for all} &j \in[n_2] \ .
        \end{matrix}
        \end{align}
        For any $p \in [0,1]$, if $E \sim \cN(p)$, then $F \sim \cN(p')$, where 
        \begin{align}\label{eq:Fstrength}
            p' = 4 \cdot p^{\frac{1}{4w}} \ .
        \end{align}
    \end{enumerate}

\begin{figure}
\begin{quantikz}
\lstick[3]{$\mathcal{A}_1$}   & \gate[6]{E}   & \qw       &\gate[3]{Z(Bz)} & \gate[3]{X(Az)}      & \qw \\ 
                            &               & \qw       &                &                      & \qw \\
                            &               & \qw       &                &                      & \qw \\ 
\lstick[3]{$\mathcal{A}_2$}   &               & \meter{$z_1$}  &\cwbend{-1}     & \cwbend{-1}     &  \\                   
                            &               & \meter{$\vdots$}  &\cwbend{-2}     & \cwbend{-2}  & \\ 
                            &               & \meter{$z_{n_2}$}  &\cwbend{-3}     & \cwbend{-3}     &
\end{quantikz}
=
\begin{quantikz}
\lstick[3]{$\mathcal{A}_1$}   & \qw       &\gate[3]{Z(Bz)} & \gate[3]{X(Az)}      & \gate[3]{F} &\qw \\ 
                            & \qw       &                &                      & \qw &\qw \\
                            & \qw       &                &                      & \qw &\qw\\ 
\lstick[3]{$\mathcal{A}_2$}   & \meter{$z_1$}  &\cwbend{-1}     & \cwbend{-1}     & \\                   
                            & \meter{$\vdots$}  &\cwbend{-2}     & \cwbend{-2}  & \\ 
                            & \meter{$z_{n_2}$}  &\cwbend{-3}     & \cwbend{-3}     &
\end{quantikz}
\caption{Two different noisy implementations~$\cQ_1$ (left) and~$\cQ_2$ (right) of an adaptive quantum circuit~$\cQ$ are equivalent if two Paulis $E \in \cP_{n}$ and $F \in \cP_{n_2}$ satisfies Eq.~\eqref{eq:FintermsofE}.\label{fig:noisyadaptivecircuitsequiv}}
\end{figure}
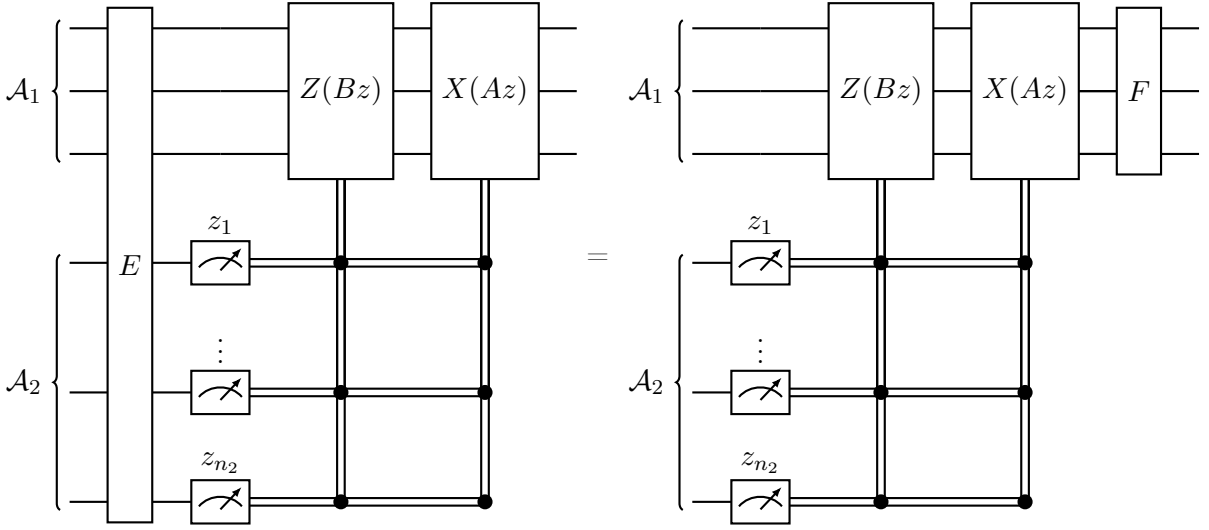

\end{lemma}
\begin{proof}
    We construct the Pauli operator~$F$ as follows.
    Let us write
    \begin{align}
        E = X(e) Z(f) \ ,
    \end{align}
    where
    \begin{align}
        e = (e^{(1)}, e^{(2)}) , \ f = (f^{(1)}, f^{(2)}) \in \mathbb{F}_2^{n} \qquad \text{with} \qquad \ e^{(j)},  \ f^{(j)} \in \mathbb{F}_2^{n_j} \quad \text{for} \quad j \in \{1,2\} \ .
    \end{align}
    Define a Pauli operator~$F \in \cP_n$ on~$\cA$ as
    \begin{align}\label{eq:FintermsofE}
    F := X(e^{(1)} + Ae^{(2)}) Z(f^{(1)} + Be^{(2)}) \ .
    \end{align}
    
    We show that the Pauli~$F$ satisfies the property~\eqref{it:noisyEFequiv}. Observe that
    \begin{align}
        \left( X(Az) Z(Bz) \otimes \bra{z}_{\cA_2} \right) E 
        &= \left(F  X(A (z+e^{(2)})) Z(B(z+e^{(2)})) \right) \otimes \bra{z+e^{(2)}}_{\cA_2}  \ , \label{eq:localstochasticEtoF}
    \end{align}
    where a phase factor from the anti-commutativity of Pauli-$X$ and Pauli-$Z$ is ignored. 
    Let~$\cQ_j(\rho)$ denote the output state of the circuit~$\cQ_j$ with an input state~$\rho$ on $\cA$ for $j \in \{1,2\}$.
    It follows from Eq.~\eqref{eq:localstochasticEtoF} that for any state~$\rho$ on the system $\cA$, we have 
    \begin{align}
        \cQ_1(\rho) &= \sum_{z \in \mathbb{F}_2^m} \left( X(Az) Z(Bz) \otimes \bra{z} \right) E \rho E^* \left( (X(A z)Z(Bz))^* \otimes \ket{z} \right) \\ 
        &= \sum_{z \in \mathbb{F}_2^m} \left(F X(A(z+e^{(2)}))Z(z+e^{(2)}) \otimes \bra{z+e^{(2)}} \right) \rho \left(h.\ c. \right) \\
        &= \cQ_2(\rho) \ ,
    \end{align}
    and thus~$\cQ_1$ and~$\cQ_2$ are equivalent.

    Next, we show that the Pauli~$F$ satisfies the property~\eqref{it:Fstrength}. 
    We first claim that
    \begin{align}
        X(Ae^{(2)}) \sim \cN(4 \cdot p^{1/w}) \qquad \text{and} \qquad Z(Be^{(e)}) \sim \cN(4 \cdot p^{1/w}) \ . \label{eq:XAelocalstochastic}
    \end{align}
    Here we only give a proof of the claim~\eqref{eq:XAelocalstochastic} for $X(Ae^{(2)})$ since the proof for $Z(Be^{(2)})$ mirrors this.
    Let~$I = \{i_1, \dots, i_\ell\} \subseteq [n_1]$ be the set of indices such that the $i$-th row vector of $A$ is non-zero, and let us denote by $j^{(i)}_1 < \dots < j^{(i)}_w$ for each $i \in I$ the indices in $[n_2]$ consisting of the support of the $i$-th row vector of $A$, i.e.,
    \begin{align}
        a_{i j^{(i)}_k} = 1 \qquad \text{for all} \qquad k \in [w], \ i \in I \ ,
    \end{align}
    and define
    \begin{align}
        g^{(k)} = \left(e^{(2)}_{j_k^{(i_1)}}, \dots, e^{(2)}_{j_k^{(i_\ell)}}\right) \in \mathbb{F}_2^{\ell} \qquad \text{for all} \qquad k \in [w] \ .
    \end{align}
    The assumption about the columns of the matrix~$A$ implies that for each $k$, the indices~$\{j_k^{(i)}\}_{i \in I} \subseteq [n_2]$ are pairwise disjoint. Therefore, there exists some permutation $\pi$ on $[n_2]$ such that $g^{(k)}$ can be obtained by truncating  the last $n_2 - \ell$ bits of the vector $\pi(e^{(2)}):=(e^{(2)}_{\pi(1)}, \dots, e^{(2)}_{\pi(n_2)})$. Therefore, by the properties~\eqref{it:subset} and~\eqref{it:permutation} of local stochastic noise in Lemma~\ref{lem:propertieslocalstochasticnoise}, we have that 
    \begin{align}
    X(g^{(k)}) \sim \cN(p) \qquad \text{for all} \qquad k \in [w] \ .
    \label{eq:Xeklocalstochastic}
    \end{align}
    By construction, we have $A e^{(2)} = g^{(1)} + \dots + g^{(w)}$, and thus
    \begin{align}
        X(Ae^{(2)}) = \prod_{k=1}^w X(g^{(k)}) \sim \cN\left(4 \cdot p^{2^{- \ceil{\log_2 w}}} \right) \ , \label{eq:XAelocalstochasticwithlog}
    \end{align}
    where we applied Lemma~\ref{lem:propertieslocalstochasticnoise}~\eqref{it:productdependent} $w-1$ times together with the fact~\eqref{eq:Xeklocalstochastic} to compute noise strength of the product of $\{X(g^{(k)})\}_{k=1}^w$ in the form of a complete binary tree. The claim~\eqref{eq:XAelocalstochastic} follows from Eq.~\eqref{eq:XAelocalstochasticwithlog} together with the inequality $4 \cdot p^{2^{- \ceil{\log_2 w}}} \leq 4 \cdot p^{1/w}$.

    Finally, we observe that $X(e^{(1)})Z(f^{(1)}) \sim \cN(p)$ by the property~\eqref{it:subset} of Lemma~\ref{lem:propertieslocalstochasticnoise}. Combining this with~\eqref{eq:XAelocalstochastic}, we have 
    \begin{align}
    F = X(Ae^{(2)}) Z(Be^{(2)}) X(e^{(1)}) Z(f^{(1)}) \sim \cN\left(4 \cdot p^{\frac{1}{4w}}\right) \ ,
    \end{align}
    where Lemma~\ref{lem:propertieslocalstochasticnoise}~\eqref{it:productdependent} is used twice and a phase factor is again ignored.
\end{proof}
An important special case of Lemma~\ref{lem:noisyadaptivecircuits} is the parallel execution of entanglement-swapping circuits.
\begin{lemma}[Entanglement swapping on noisy Bell states]\label{lem:noisyentanglementswapping}
    Consider the following entanglement swapping protocol which takes $k\ell$ Bell pairs
    \begin{align}
        \Psi_{\mathrm{in}} = \bigotimes_{i=1}^{k} \bigotimes_{j=1}^\ell \Phi_{Q_i^{(j)}{Q'_i}^{(j)}}
    \end{align}
    as an input and generates $\ell$ Bell pairs~$\Psi_{\mathrm{out}} = \bigotimes_{j=1}^{\ell} \Phi_{Q_1^{(j)} {Q'_k}^{(j)}}$:
    \begin{enumerate}
        \item Apply a Bell measurement to the qubits~${Q'_i}^{(j)}, Q_{i+1}^{(j)}$ for each $i \in [k-1]$ and $j \in [\ell]$, getting outcomes $(a_i^{(j)}, b_i^{(j)})$. \label{it:entanglementswapbellmeas}
        \item Apply $X^{a^{(j)}}Z^{b^{(j)}}$ to qubit~${Q'_k}^{(j)}$ for all $j \in [\ell]$ where $a^{(j)} = \bigoplus_{i=1}^{k-1} a_i^{(j)}$ and $b^{(j)} = \bigoplus_{i=1}^{k-1} b_i^{(j)}$.
    \end{enumerate}
    Let $E \sim \cN(p)$ be a~$2k\ell$-qubit Pauli noise which is local stochastic of strength~$p$.
    If the entanglement swapping protocol takes noisy $k\ell$ Bell pairs $E \Psi_{\mathrm{in}}$, then the output state is a noisy $\ell$-fold Bell state $F\Psi_{\mathrm{out}}$
    for some $2\ell$-qubit local stochastic noise $F \sim \cN(p')$, where $p' = 4 (\sqrt{2p})^{\frac{1}{4(k-1)}}$.
\end{lemma}
\begin{proof}
    Note that Step~\ref{it:entanglementswapbellmeas} can be decomposed into one layer~$U$ consisting of two-qubit Clifford unitaries followed by single-qubit measurements in computational basis.
    By commuting~$E$ past the layer~$U$, we see that 
    the protocol with state~$E \Psi_{\mathrm{in}}$ is equivalent to the following protocol with the input state~$(U E U^\dagger)U\bigotimes_{i=1}^{k} \bigotimes_{j=1}^\ell \Phi_{Q_i^{(j)}{Q'_i}^{(j)}}$:
    \begin{enumerate}
        \item Apply computational basis measurement to the qubits~${Q'_i}^{(j)}, Q_{i+1}^{(j)}$ for each $i \in [k-1]$ and $j \in [\ell]$, getting outcomes $(a_i^{(j)}, b_i^{(j)})$.
        \item Apply $X^{a^{(j)}}Z^{b^{(j)}}$ to qubit~${Q'_k}^{(j)}$ for all $j \in [\ell]$ where $a^{(j)} = \bigoplus_{i=1}^{k-1} a_i^{(j)}$ and $b^{(j)} = \bigoplus_{i=1}^{k-1} b_i^{(j)}$.
    \end{enumerate}
    Note that this has the form of an adaptive circuit as discussed in Lemma~\ref{lem:noisyadaptivecircuits}.
    By Lemma~\ref{lem:noisyadaptivecircuits}, the resulting state is a noisy~$\ell$-fold Bell state affected by a Pauli noise~$F$ which is determined by~$UEU^\dagger$.
    By Lemma~\ref{lem:propertieslocalstochasticnoise}~\eqref{it:clifford}, we have
    \begin{align}\label{eq:equivalentinputerrorstrength}
        UEU^\dagger \sim \cN\left(\sqrt{2p}\right) \ .
    \end{align}
    
    Let us label the qubits by the sets~$\cA_1$ and~$\cA_2$ defined as
    \begin{align}
        \cA_1 &= \bigcup_{j=1}^\ell \{Q_1^{(j)}, {Q'_k}^{(j)}\} = [2\ell] \\
        \cA_2 &= \bigcup_{i=1}^{k-1} \bigcup_{j=1}^\ell \{{Q'_i}^{(j)}, Q_{i+1}^{(j)}\} = \{2\ell+1, 2\ell+2, \dots,   2k\}  \ ,
    \end{align}
    and construct a string~$z = (z_1, \dots, z_{(2k-2) \ell}) \in \mathbb{F}_2^{(2k-2) \ell}$ by collecting all bits
    \begin{align}
       \bigcup_{i=1}^{k-1} \bigcup_{j=1}^{\ell}  \{(a_i^{(j)}, b_i^{(j)})\} \ , 
    \end{align}
    where the ordering is consistent with the labeling~$\cA_2$.
    Then, it is clear by inspection that
    there are matrices~$A, B \in \mathbb{F}_2^{2\ell \times (2k-2) \ell}$ satisfying 
    \begin{align}
        \bigotimes_{j=1}^{\ell} X^{a^{(j)}}_{{Q'_k}^{(j)}} Z^{b^{(j)}}_{{Q'_k}^{(j)}} = X(Az)Z(Bz) \ .
    \end{align}
    Furthermore, all non-zero row vectors of $A$ and $B$ have weight~$k-1$ and all column vectors of $A$ and $B$ have weight at most 1. As a result, we can apply Lemma~\ref{lem:noisyadaptivecircuits}~\eqref{it:Fstrength}.
    We obtain $p'$ by inserting~$\sqrt{2p}$, which was obtained by Eq.~\eqref{eq:equivalentinputerrorstrength}, and $k-1$ respectively into $p$ and $w$ of Eq.~\eqref{eq:Fstrength}.
\end{proof}
Another consequence of Lemma~\ref{lem:noisyadaptivecircuits} is the following statement, which characterizes the effect of local stochastic noise on  teleportation circuit applied in parallel.

\begin{lemma}[Noisy implementation of the quantum teleportation protocol] \label{lem:noisyquantumteleportation}
    Consider the following quantum teleportation protocol which takes a $3\ell$-qubit state
    \begin{align}
        \Psi'_{\mathrm{in}} = \Psi_{Q_1 \dots Q_\ell} \bigotimes_{j=1}^{\ell} \Phi_{R_j R'_j}
    \end{align}
    as an input and outputs an~$\ell$-qubit state~$\Psi'_{\mathrm{out}}=\Psi_{R'_1 \dots R'_\ell}$:
    \begin{enumerate}
        \item Apply a Bell measurement to the qubits $Q_j, R_j$ for each $j \in [\ell]$, getting outcomes~$(a_j, b_j)$.
        \item Apply $X^{a_j} Z^{b_j}$ to the qubit $R'_j$ for all~$j \in [\ell]$.
    \end{enumerate}
    Let $E \sim \cN(p)$ be a $3\ell$-qubit Pauli noise which is local stochastic noise of strength $p$. If the quantum teleportation protocol takes a noisy state~$E \Psi'_{\mathrm{in}}$ as an input, then the output state is a noisy state~$F \Psi'_{\mathrm{out}}$ for some $\ell$-qubit local stochastic noise $F \sim \cN(p')$, where $p' = 2^{\frac{17}{8}} p^{\frac{1}{8}}$.
\end{lemma}
\begin{proof}
    The proof mirrors that of Lemma~\ref{lem:noisyentanglementswapping}, so we only give a sketch. 
    We choose~$\sqrt{2p}$ and $1$ respectively for the parameters $p$ and $w$ of Eq.~\eqref{eq:Fstrength} to obtain $p'$.
\end{proof}

\end{document}